\documentclass[twocolumn]{article}

\usepackage{amsmath}
\usepackage{amssymb}
\usepackage{amsthm} 
\usepackage{xspace}
\usepackage[usenames,dvipsnames]{xcolor}
\usepackage[arrow,matrix,curve,color]{xypic}
\usepackage{subfig}
\usepackage{etoolbox}
\usepackage[normalem]{ulem}
\usepackage{tikz}

\theoremstyle{plain}
\newtheorem{theorem}{Theorem}[section]
\newtheorem{proposition}[theorem]{Proposition}
\newtheorem{lemma}[theorem]{Lemma}

\theoremstyle{definition}
\newtheorem{definition}{Definition}[section]

\theoremstyle{remark}
\newtheorem*{remark}{Remark}

\newcommand{\eat}[1]{}

\newcommand{\Small}{\fontsize{7}{9}\selectfont}

\makeatletter
\let\@oplus\oplus
\renewcommand{\oplus}{\raisebox{0.75pt}{\ensuremath{\boldsymbol\@oplus}} }
\makeatother
\robustify{\oplus}

\newcommand{\xrsolution}{XR-solution}

\newcommand{\cqa}{\ensuremath{\textup{CQA}}\xspace}
\newcommand{\certain}{\ensuremath{\textup{certain}}\xspace}
\newcommand{\xrcertain}{\ensuremath{\textup{XR-certain}}\xspace}

\newcommand{\ARcertain}{\ensuremath{\textup{AR-certain}}\xspace}

\newcommand{\gav}{\textsc{gav}}
\newcommand{\lav}{\textsc{lav}}
\newcommand{\glav}{\textsc{glav}}
\newcommand{\waglav}{\textsc{waglav}}
\newcommand{\efsotgd}{\textsc{efsotgd}}
\newcommand{\sotgd}{\textsc{sotgd}}
\newcommand{\egd}{\textsc{egd}}

\newcommand{\ptime}{$\mathrm{PTIME}$\xspace}
\newcommand{\NP}{$\mathrm{NP}$\xspace}
\newcommand{\coNP}{$\mathrm{coNP}$\xspace}

\newcommand{\sm}[1]{\ensuremath{\mathcal{#1}}\xspace} 
\newcommand{\mapping}{\ensuremath{\mathcal{M}=(\mathbf{S},\mathbf{T},\sigst,\sigt)}\xspace} 
\newcommand{\sigs}{\ensuremath{\Sigma_{\text{s}}}\xspace} 
\newcommand{\sigst}{\ensuremath{\Sigma_{\text{st}}}\xspace} 
\newcommand{\sigt}{\ensuremath{\Sigma_{\text{t}}}\xspace} 

\newcommand{\mappingtypetwo}[2]{\textup{\sc #1+#2}} 
\newcommand{\mappingtypethree}[3]{\textup{\sc #1+(#2,\,#3)}} 

\newcommand{\rel}[1]{\ensuremath{\mathrm{#1}}\xspace}
\DeclareMathOperator{\Nulls}{Nulls}
\DeclareMathOperator{\Const}{Const}

\newcommand{\dlp}[1]{\ensuremath{\Pi_{\textsc{xr}c}(#1)}}

\newcommand\restr[2]{{
  \left.\kern-\nulldelimiterspace 
  #1 
  \vphantom{\big|} 
  \right|_{#2} 
  }}

\newcommand{\defeq}{\ensuremath{\text{ := }}}

\newcommand{\mds}{\mathrm {mod}}
\newcommand{\drep}{\mathrm {drep}}

\makeatletter
\def\blfootnote{\xdef\@thefnmark{}\@footnotetext}
\makeatother

\usepackage{enumitem}
\setitemize{topsep=0pt,parsep=0pt,partopsep=0pt, leftmargin=13pt}
\setenumerate{topsep=0pt,parsep=0pt,partopsep=0pt, leftmargin=16pt}
\setdescription{topsep=0pt,parsep=0pt,partopsep=0pt, leftmargin=13pt}
\begin{document}

\title{Exchange-Repairs:\\Managing Inconsistency in Data Exchange}

\author{Balder ten Cate     \\ {\tt btencate@ucsc.edu} \\ UC Santa Cruz \\ Google,\and 
		Richard L. Halpert  \\ {\tt rhalpert@ucsc.edu} \\ UC Santa Cruz \and 
		Phokion G. Kolaitis \\ {\tt kolaitis@ucsc.edu} \\ UC Santa Cruz \\ IBM Research - Almaden}


\maketitle

\begin{abstract}

In a data exchange setting with target constraints, it is often the case that a given source instance has no solutions. Intuitively, this happens
when data sources contain inconsistent or conflicting information that
is exposed by the target constraints at hand.
In such cases, the semantics of target queries trivialize, because the certain answers of every target query over the given source instance evaluate to ``true". The aim of this paper is to introduce and explore a new framework that gives meaningful semantics in such cases by using the notion of exchange-repairs. Informally, an exchange-repair of a source instance is another source instance that differs minimally from the first, but has a solution. In turn, exchange-repairs give rise to a natural notion of exchange-repair certain answers (in short, XR-certain answers) for target queries in the context of data exchange with target constraints.

After exploring the structural properties of exchange-repairs, we focus on the problem of computing the XR-certain answers of conjunctive queries. We show that for schema mappings specified by source-to-target GAV dependencies and target equality-generating dependencies (egds), the XR-certain answers of a target conjunctive query can be rewritten as the consistent answers (in the sense of standard database repairs) of a union of conjunctive queries over the source schema with respect to a set of egds over the source schema, thus making it possible to use a consistent query-answering system to compute XR-certain answers in data exchange. In contrast, we show that this type of rewriting is not possible for schema mappings specified by source-to-target LAV dependencies and target egds, nor for schema mappings specified by both source-to-target and target GAV dependencies. We then examine the general case of schema mappings specified by source-to-target GLAV constraints, a weakly acyclic set of target tgds and a set of target egds.  The main result asserts that, for such settings, the XR-certain answers of conjunctive queries can be rewritten as the certain answers of a union of conjunctive queries with respect to the stable models of a disjunctive logic program over a suitable expansion of the source schema.

\end{abstract}

\section{Introduction and Summary of Contributions}

Data exchange is the problem of transforming data structured under one schema, called the source schema, into data structured under a different schema, called the target schema, in such a way that pre-specified constraints on these two schemas are satisfied. Data exchange is a ubiquitous data inter-operability task that has been explored in depth during the past decade (see \cite{DBLP:series/synthesis/2010Arenas}).
 This task is formalized with the aid of  schema mappings \mapping{}, where   $\bf S$ is the source schema, $\bf T$ is the target schema, \sigst{} is a set of constraints between $\bf S$ and $\bf T$,  and \sigt{} is a set of  constraints on $\bf T$. The most thoroughly investigated schema mappings are the ones in which \sigst{} is a set of source-to-target tuple-generating dependencies (s-t tgds) and \sigt{} is a set of target tuple-generating dependencies (target tgds) and target equality-generating dependencies (target egds) \cite{DBLP:journals/tcs/FaginKMP05}.  An example of such a schema mapping, along with a target query, follows:

 	\begin{figure*}
		\centering\small
		\begin{tabular}{lll}\\
		\hspace{-1mm}\sigst & $=$ & $\left\{ \hspace{-1mm}\begin{array}{lllll} {\tt Task\_Assignments}(p,t,d) & \hspace{-1mm}\to & {\tt Departments}(p,d) \hspace{-.5mm}\wedge\hspace{-.5mm} {\tt Tasks}(p,t)\hspace{-1mm} \\
															{\tt Stakeholders\_old}(t,s) & \hspace{-1mm}\to & {\tt Stakeholders\_new}(t,s) \end{array} \right\}$ \\
		\hspace{-1mm}\sigt & $=$ & $\left\{ \hspace{-0.5mm}\begin{array}{l} {\tt Departments}(p,d) \wedge {\tt Departments}(p,d') \to d=d' \end{array} \right\}$\\
		~\\\hline
		~\\
		\multicolumn{3}{l}{${\tt boss}(person,stakeholder)=\exists task. $}\\
		\multicolumn{3}{r}{${\tt Tasks}(person,task) \wedge {\tt Stakeholders\_new}(task,stakeholder)$}\\
		\end{tabular}
		\caption{A schema mapping \sm{M} specified by tgds and egds, and a target query.  In this example, the egd is actually a key constraint and there are no target tgds.}
		\label{example:mapping}
	\end{figure*}

   Every schema mapping  \mapping{} gives rise to two distinct algorithmic problems. The first is the existence and construction of solutions: given a source instance $I$,  determine whether  a \emph{solution} for $I$ exists (i.e., a target instance $J$ so that $(I,J)$ satisfies $\Sigma_{st} \cup \Sigma_t$) and, if it does,  construct such a ``good" solution. The second  is to compute the \emph{certain answers} of target queries, where if $q$ is a target query and $I$ is a source instance, then  $\certain(q,I,{\cal M})$ is the intersection
  of the sets $q(J)$, as $J$ varies over all solutions for $I$. For arbitrary schema mappings specified by s-t tgds and target tgds and egds,  both these problems can be undecidable \cite{DBLP:conf/pods/KolaitisPT06}. However, as shown in \cite{DBLP:journals/tcs/FaginKMP05}, if the set \sigt{} of target tgds obeys a mild structural condition, called \emph{weak acyclicity}, then both these problems can be  solved in polynomial time using the \emph{chase procedure}. Given a source instance $I$, the chase procedure attempts to build a ``most general" solution $J$ for $I$  by generating facts that satisfy each s-t tgd and each target tgd as needed, and by equating two nulls or equating a null to a constant, as dictated by the egds. If the chase procedure encounters an egd that equates two distinct constants, then it terminates and reports that no solution for $I$ exists. Otherwise, it constructs a \emph{universal} solution $J$ for $I$, which can also be used to compute the certain answers of conjunctive queries in time bounded by a polynomial in the size of $I$.

Consider the situation in which the chase terminates and reports that no solution exists.
In such cases,  for every boolean target query $q$, the certain answers
$\certain(q,I,{\cal M})$ evaluate to ``true". Even though the certain answers have become the standard semantics of queries in the data exchange context, there is clearly something unsatisfactory about this state of affairs, since the certain answers trivialize when no solutions exist.  Intuitively, the root cause for the lack of solutions is that the source instance contains inconsistent or conflicting information that
is exposed by the target constraints of the schema mapping at hand. In turn, this suggests that alternative semantics for target queries could be obtained by adopting the notions of database \emph{repairs} and \emph{consistent answers} from the study of inconsistent databases (see
\cite{DBLP:series/synthesis/2011Bertossi} for an overview).
We note that several different types of repairs have been studied in the context of inconsistent databases; the most widely used ones are the \emph{symmetric difference ($\oplus$-repairs)}, which contain as special cases the
\emph{subset-repairs} and the \emph{superset-repairs}.

How can the notions of database repairs and consistent answers be adapted to the data exchange framework?
When one reflects on this question, then one realizes that several different approaches are possible.

One approach, which we  call {\em materialize-then-repair}, is as follows: given a source instance,  a target instance is produced by chasing with the source-to-target tgds in $\sigst$ and the target tgds in $\sigt$, while ignoring the target egds in $\sigt$. Since the target instance produced this way may very well violate the egds in $\sigt$, it is treated as an inconsistent instance w.r.t.\ $\sigt$; consider its repairs. Note that a similar approach has been adopted by \cite{DBLP:conf/fqas/BertossiCCG02,DBLP:conf/ijcai/BravoB03} in the context of data integration.
A different approach, which we call  {\em exchange-as-repair}, treats the given source instance as an inconsistent instance over the combined schema ${\bf S} \cup \bf T$ w.r.t.\ the union  $\sigst \cup \sigt$  and considers its repairs. Note that this is in the spirit of~\cite{DBLP:conf/icdt/GrahneO10}, where instances in peer data exchange that do not satisfy the schema mapping at hand are treated as inconsistent databases over a combined schema.
We now point out that neither of these approaches gives rise to satisfactory semantics.

	\begin{figure}
		\Small
		\begin{tabular}{ccrl}
			\multicolumn{1}{c}{$I$} &  & \multicolumn{2}{c}{$J$}\\
			\cline{1-1}\cline{3-4}
			~\\
			\begin{tabular}{|c|c|c|}
				\hline
				\multicolumn{3}{|c|}{\tt Task\_Assignments} \\
				person & task & dept \\
				\hline
				peter & tpsreport & software \\
				peter & spaceout & software \\
				peter & meetbobs & exec \\
				\hline
			\end{tabular} & &
			\begin{tabular}{|c|c|}
				\hline
				\multicolumn{2}{|c|}{\tt Departments} \\
				person & dept \\
				\hline
				peter & software \\
				peter & exec \\
				~ & \\
				\hline
			\end{tabular} &
			\begin{tabular}{|c|c|}
				\hline
				\multicolumn{2}{|c|}{\tt Tasks} \\
				person & task \\
				\hline
				peter & tpsreport \\
				peter & spaceout \\
				peter & meetbobs \\
				\hline
			\end{tabular}\\
			\\
			\begin{tabular}{|c|c|}
				\hline
				\multicolumn{2}{|c|}{\tt Stakeholders\_old} \\
				task & stakeholder \\
				\hline
				tpsreport & lumbergh \\
				tpsreport & portman \\
				spaceout & bobs \\
				meetbobs & bobs \\
				\hline
			\end{tabular} & &
			\multicolumn{2}{c}{
			\begin{tabular}{|c|c|}
				\hline
				\multicolumn{2}{|c|}{\tt Stakeholders\_new} \\
				task & stakeholder \\
				\hline
				tpsreport & lumbergh \\
				tpsreport & portman \\
				spaceout & bobs \\
				meetbobs & bobs \\
				\hline
			\end{tabular}
			}\\
		\end{tabular}
		\caption{A source instance $I$ and the inconsistent target instance $J$ that results from chasing $I$ with the tgds in \sm{M}.}
		\label{fig:inconsistent_chase}
	\end{figure}

	\begin{figure}
		\Small\centering
		\begin{tabular}{rl}
			\multicolumn{2}{c}{$J'$}\\
			\cline{1-2}
			\\
			\begin{tabular}{|c|c|}
				\hline
				\multicolumn{2}{|c|}{\tt Departments} \\
				person & dept \\
				\hline
				\sout{peter} & \sout{software} \\
				peter & exec \\
				~ & \\
				\hline
			\end{tabular} &
			\begin{tabular}{|c|c|}
				\hline
				\multicolumn{2}{|c|}{\tt Tasks} \\
				person & task \\
				\hline
				peter & tpsreport \\
				peter & spaceout \\
				peter & meetbobs \\
				\hline
			\end{tabular} \\
			\\
			\multicolumn{2}{c}{
			\begin{tabular}{|c|c|}
				\hline
				\multicolumn{2}{|c|}{\tt Stakeholders\_new} \\
				task & stakeholder \\
				\hline
				tpsreport & lumbergh \\
				tpsreport & portman \\
				spaceout & bobs \\
				meetbobs & bobs \\
				\hline
			\end{tabular}
			}\\
		\end{tabular}
		\caption{A symmetric-difference-repair of $J$ w.r.t. \sigt{}.}
		\label{fig:exchange_then_repair}
	\end{figure}

Figure~\ref{fig:inconsistent_chase} gives an example of a target instance that is produced in the materialize-then-repair approach by chasing with the s-t tgds in Figure~\ref{example:mapping}.
 Clearly, $J$ is inconsistent because it violates the egd in \sigt{}.
Consider now the subset repair $J'$ in Figure \ref{fig:exchange_then_repair} of our materialized target instance $J$ (note that, in this case, symmetric difference repairs coincide with subset repairs).
Notice that the repair $J'$ places peter in the exec department, yet still has him performing tasks for the software department -- the fact that the ``tpsreport'' and ``spaceout'' tasks are derived from a tuple placing peter in the software department has been lost.  The only other repair of $J$ similarly fails to reflect the shared origin of tuples in the {\tt Tasks} and {\tt Departments} tables, and this disconnect in the materialize-then-repair approach manifests in the consistent answers to target queries.  In this example, the consistent answers for ${\tt boss}({\rm peter},b)$ are $\{({\rm peter},{\rm bobs}), ({\rm peter},{\rm portman}), ({\rm peter},{\rm lumbergh})\}$.  However, the last two tuples are derived from facts placing peter in the software department, even though in $J'$ he is not.

	\begin{figure*}[htbp]\Small\centering
		\begin{tabular}{rlcrlcrl}
			\multicolumn{2}{c}{$(I_1,J_1)$} & & \multicolumn{2}{c}{$(I_2,J_2)$} & & \multicolumn{2}{c}{$(I_3,J_3)$}\\
			\cline{1-2} \cline{4-5} \cline{7-8}
			\\
			\multicolumn{2}{c}{
			\begin{tabular}{|c|c|c|}
				\hline
				\multicolumn{3}{|c|}{\tt Task\_Assignments} \\
				person & task & dept \\
				\hline
				\sout{peter} & \sout{tpsreport} & \sout{software} \\
				\sout{peter} & \sout{spaceout} & \sout{software} \\
				peter & meetbobs & exec \\
				\hline
			\end{tabular}} & &
			\multicolumn{2}{c}{
			\begin{tabular}{|c|c|c|}
				\hline
				\multicolumn{3}{|c|}{\tt Task\_Assignments} \\
				person & task & dept \\
				\hline
				peter & tpsreport & software \\
				peter & spaceout & software \\
				\sout{peter} & \sout{meetbobs} & \sout{exec} \\
				\hline
			\end{tabular}} & &
			\multicolumn{2}{c}{
			\begin{tabular}{|c|c|c|}
				\hline
				\multicolumn{3}{|c|}{\tt Task\_Assignments} \\
				person & task & dept \\
				\hline
				peter & tpsreport & software \\
				\sout{peter} & \sout{spaceout} & \sout{software} \\
				\sout{peter} & \sout{meetbobs} & \sout{exec} \\
				\hline
			\end{tabular}}
			\\
			\\
			\begin{tabular}{|c|c|}
				\hline
				\multicolumn{2}{|c|}{\tt Departments} \\
				person & dept \\
				\hline
				~ & \\
				peter & exec \\
				~ & \\
				\hline
			\end{tabular} &
			\begin{tabular}{|c|c|}
				\hline
				\multicolumn{2}{|c|}{\tt Tasks} \\
				person & task \\
				\hline
				~ & \\
				~ & \\
				peter & meetbobs \\
				\hline
			\end{tabular} & &
			\begin{tabular}{|c|c|}
				\hline
				\multicolumn{2}{|c|}{\tt Departments} \\
				person & dept \\
				\hline
				peter & software \\
				~ & \\
				~ & \\
				\hline
			\end{tabular} &
			\begin{tabular}{|c|c|}
				\hline
				\multicolumn{2}{|c|}{\tt Tasks} \\
				person & task \\
				\hline
				peter & tpsreport \\
				peter & spaceout \\
				~ & \\
				\hline
			\end{tabular} & &
			\begin{tabular}{|c|c|}
				\hline
				\multicolumn{2}{|c|}{\tt Departments} \\
				person & dept \\
				\hline
				peter & software \\
				~ & \\
				~ & \\
				\hline
			\end{tabular} &
			\begin{tabular}{|c|c|}
				\hline
				\multicolumn{2}{|c|}{\tt Tasks} \\
				person & task \\
				\hline
				peter & tpsreport \\
				~ & \\
				~ & \\
				\hline
			\end{tabular}\\			
		\end{tabular}
		\caption{Three repairs of $(I,\emptyset)$ (from Figure~\ref{fig:inconsistent_chase}) w.r.t. $\sigst \cup \sigt$.  The {\tt Stakeholders} tables are omitted for brevity.}
		\label{fig:exchange_as_repair}
	\end{figure*}

	The situation is no better in the exchange-as-repair approach.
Figure~\ref{fig:exchange_as_repair} depicts three repairs of this type (using symmetric difference semantics).
While the first two repairs in Figure~\ref{fig:exchange_as_repair} seem reasonable, in the third we have eliminated {\tt Task\_Assignments}(peter, spaceout, software), even though our key constraint is already satisfied by the removal of {\tt Task\_Assignments}(peter, meetbobs, exec) alone.  In this approach, the consistent answers of ${\tt boss}({\rm peter},b)$ are $\emptyset$, despite the intuitive conclusion that peter should be performing tasks for the bobs regardless of which way we fix the department key constraint violation.  For symmetric-difference repairs, it is equally valid to satisfy a violated tgd by removing tuples as by adding them\footnote{A noteworthy alternative to symmetric difference repairs are the loosely-sound semantics of~\cite{DBLP:conf/ijcai/CaliLR03}, discussed in detail in Section~\ref{data-integration-connections}.}.  However, in a data exchange setting, the target instance is initially empty, so it would be more natural to satisfy violated tgds by deriving new tuples.  This observation motivates the particulars of our approach, which we introduce next.

\subsection{Summary of Contributions}
  Our aim in this paper is to introduce and explore a new framework that gives meaningful and non-trivial semantics to queries in data exchange, including cases in which no solutions exist for a given source instance.
   
   At the conceptual level, the main contribution is the introduction of the notion of an \emph{exchange-repair}.
    Informally, an exchange-repair of a source instance is another source instance that differs minimally from the first, but has a solution. Exchange-repairs give rise to a natural notion of \emph{exchange-repair certain answers} (in short, \emph{XR-certain answers}) for target queries in the context of data exchange.
  Note that  if a source instance $I$ has a solution, then the XR-certain answers of target queries on $I$ coincide with the certain answers of the queries on $I$. If $I$ has no solutions, then unlike the certain answers, the XR-certain answers are non-trivial and meaningful.
   
   We provide examples  demonstrating that these new semantics improve upon both the materialize-then-repair approach and the exchange-as-repair approach discussed earlier.
 We also produce a detailed comparison of the XR-certain semantics with the main notions of inconsistency-tolerant semantics studied in data integration and in ontology-based data access. This comparison is carried out in Section \ref{sec:related}, after we have introduced our framework and presented some basic structural properties of exchange-repairs in Section \ref{sec:xr}.


After this, we focus on the problem of computing the XR-certain answers of conjunctive queries. In Section \ref{sec:CQA-rewritability}, we show that for schema mappings specified by source-to-target GAV (global-as-view) dependencies and target egds, the XR-certain answers of  conjunctive queries can be rewritten as the consistent answers (in the sense of standard database repairs) of a union of conjunctive queries over the source schema with respect to a set of egds over the source schema, thus making it possible to use a consistent query-answering system to compute XR-certain answers in data exchange. In contrast, we show that this type of rewriting is not possible for schema mappings specified by source-to-target LAV (local-as-view) dependencies and target egds, nor for schema mappings specified by source-to-target and target GAV dependencies and target egds.
 
In Section \ref{sec:dlp}, we examine the general case of schema mappings specified by s-t tgds, a weakly acyclic set of target tgds and a set of target egds.  The main result is that, for such settings, the XR-certain answers of conjunctive queries can be rewritten as the certain answers of a union of conjunctive queries with respect to the stable models of a disjunctive logic program over a suitable expansion of the source schema.  This is achieved in two steps.
 First, for schema mappings consisting of GAV s-t tgds, GAV target tgds, and target egds, we show that the XR-certain answers of conjunctive queries can be 
  reduced to cautious reasoning over stable models of a disjunctive logic program. Second, for schema mappings consisting of GLAV s-t tgds, weakly acyclic sets of GLAV target tgds, and target egds, we show that the XR-certain answers of conjunctive queries can be rewritten
  as the XR-certain answers of conjunctive queries w.r.t.\  a schema mapping consisting of GAV s-t tgds, GAV target tgds, and target egds.  In fact, we prove the stronger result that such a rewriting is possible for schema mappings specified by a second-order s-t tgd, a weakly acyclic second-order target tgd, and set of target egds.
\section{Preliminaries}
This section contains definitions of basic notions and a minimum amount of background material. Detailed information about schema mappings and certain answers can be found in \cite{DBLP:series/synthesis/2010Arenas,DBLP:journals/tcs/FaginKMP05}, and about repairs and consistent answers in \cite{DBLP:conf/pods/ArenasBC99,DBLP:series/synthesis/2011Bertossi}.

\smallskip

\subsection{Instances and Homomorphisms}
Fix an infinite set $\Const$ of elements, and an infinite set $\Nulls$ of elements such that $\Const$ and $\Nulls$ are disjoint.
A \emph{schema}  $\bf R$ is a finite set of relation symbols, each having a designated arity. 
An {\em ${\bf R}$-instance} is  a finite
database $I$ over the schema $\bf R$ whose active domain is a subset of $\Const \cup \Nulls$.
A \emph{fact} of an $\bf R$-instance $I$ is an expression of the form $R(a_1,\ldots,a_k)$, where $R$ is a relation symbol of arity $k$ in $\bf R$ and $(a_1,\ldots,a_k)$ is a member of the relation $R^I$  on $I$ that interprets the relation symbol $R$. Every $\bf R$-instance can be identified with the set of its facts.  We say that an $\bf R$-instance $I'$ is a \emph{sub-instance} of an $\bf R$-instance $I$ if $I'\subseteq I$, where $I'$ and $I$ are viewed as sets of facts.

By a {\em homomorphism} between two instances $K$ and $K'$, we mean a map from the active domain of $K$ to the active domain of $K'$ that is the identity function on all elements of $\Const$ and such that for every atom $\rel R(v_1,...,v_n) \in K$ we have that $\rel R(h(v_1),...,h(v_n)) \in K'$.

\subsection{Schema Mappings and Certain Answers.}
A \emph{tuple-generating dependency (tgd)} is an expression of the form
		$\forall {\tt\bf x}( \phi ({\tt\bf x}) \to \exists {\tt\bf y} \psi ({\tt\bf x},{\tt\bf y}))$,
where $\phi({\tt\bf x})$ and $\psi ({\tt\bf x},{\tt\bf y})$ are conjunctions of atoms over some relational schema.

Tgds are also known as GLAV (global-and-local-as-view) constraints. Tgds with no existentially quantified variables are called {\em full}. Two important special cases are the GAV constraints and the LAV constraints: the former are the tgds of the form $\forall {\tt\bf x}( \phi ({\tt\bf x}) \to P({\bf x}))$ and the latter are the tgds of the form $\forall {\tt\bf x}( R({\bf x}) \to \exists {\tt\bf y} \psi ({\tt\bf x},{\tt\bf y}))$, where $P$ and $R$ are individual relation symbols.  Every full tgd is logically equivalent to a set of \gav{} tgds that can be computed in linear time.

Suppose we have two disjoint relational schemas $\bf S$ and $\bf T$, called the {\em source} schema and the {\em target} schema.  A {\em source-to-target tgd} ({\em s-t tgd}) is a tgd as above such that $\phi({\bf x})$ is a conjunction over $\bf S$ and $\psi({\bf x},{\bf y})$ is a conjunction over $\bf T$.  When the schemas are understood from context, we may say just tgd even if the constraint is source-to-target.

An \emph{equality-generating dependency (egd)} is an expression of the form
		$ \forall {\tt\bf x}( \phi ({\tt\bf x}) \to x_i = x_j) $
with $\phi ({\tt\bf x}) $  a conjunction of atoms over a relational schema.

For the sake of readability, we will frequently drop universal quantifiers when writing tgds and egds.


A \emph{schema mapping} is a quadruple \mapping{}, where  $\bf S$ is a source schema, $\bf T$ is a target schema,
\sigst{} is a finite set of source-to-target constraints, and \sigt{} is a finite set of constraints over the target schema.

We will use the notation \glav, \gav, \lav, \egd{} to denote the classes of sets of constraints consisting of finite sets of, respectively, GLAV constraints, GAV constraints, LAV constraints, and egds.
 If $C$ is a class of sets of source-to-target dependencies and $D$ is a class of sets of target dependencies, then the notation \mappingtypetwo{$C$}{$D$}  denotes the class of all schema mappings \mapping{} such that \sigst{} is a member of $C$ and $\sigt$ is a member of $D$. For example,
  \mappingtypetwo{\glav}{\egd} denotes the class of all schema mappings \mapping{} such that \sigst{} is a finite set of s-t tgds and \sigt{} is a finite set of egds.
   Moreover, we will use the notation $(D_1,D_2)$ to denote that  the union of two classes $D_1$ and $D_2$ of sets of target dependencies. For example, \mappingtypethree{\gav}{\gav}{\egd}
     denotes the class of all schema mappings \mapping{} such that \sigst{} is a set of GAV s-t tgds and \sigt{} is the union of a finite set of GAV target tgds with a finite set of target egds.

Let \mapping{} be a schema mapping.
A target instance $J$ is a \emph{solution} for a source instance $I$ w.r.t.\ ${\mathcal M}$
if $J$ is finite, and the pair $(I,J)$ satisfies $\cal M$, i.e., $I$ and $J$ together satisfy \sigst{}, and $J$ satisfies \sigt{}.  Recall that, by definition, instances are finite.  Additionally, by convention, we will assume that source instances do not contain null values.
A \emph{universal} solution for $I$ is a solution $J$ for $I$ such that if $J'$ is a solution for $I$, then there is a homomorphism $h$ from $J$ to $J'$ that is the identity on the active domain of $I$.
If \mapping{} is an arbitrary schema mapping, then a given source instance may have no solution or it may have a solution, but no (finite) universal solution.  However, if \sigt{} is the union of a \emph{weakly acyclic} set of target tgds and a set of egds, then a solution exists if and only if a universal solution exists. Moreover, the \emph{chase procedure} can be used to determine if, given a source instance $I$, a solution for $I$ exists and, if it does, to actually construct a universal solution $chase(I,{\cal M})$ for $I$ in time polynomial in the size of $I$ (see \cite{DBLP:journals/tcs/FaginKMP05} for details).  The definition of weak acyclicity is given next, followed by the definition of the {\em chase procedure}.

\begin{definition}[\cite{DBLP:journals/tcs/FaginKMP05}]
Let $\Sigma$ be a set of tgds over a schema $\bf T$. Construct
a directed graph, called the {\em dependency graph}, as follows:

\begin{itemize}
\item Nodes: For every pair ($R$, $A$) with $R$ a relation symbol
in $\bf T$ and $A$ an attribute of $R$, there is a distinct node;
call such a pair ($R$, $A$) a {\em position}.
\item Edges: For every tgd $\forall {\bf x}(\phi({\bf x}) \rightarrow \exists {\bf
y}\, \psi({\bf x}, {\bf y}))$ in $\Sigma$ and for every $x$ in $\bf
x$ that  occurs in $\psi$, and  for every occurrence of $x$ in
$\phi$ in position ($R$, $A_i$):
\begin{enumerate}
\item  For every occurrence of $x$ in $\psi$ in position ($S$,
$B_j$), add an edge $(R, A_i) \rightarrow (S, B_j)$ (if it does
not already exist).
\item For every existentially
quantified variable $y$ and for every occurrence of $y$ in $\psi$
in position ($T$, $C_k$), add a {\em special edge} $(R, A_i)
\rightarrow (T, C_k)$ (if it does not already exist).
\end{enumerate}
\end{itemize}

 We say that   $\Sigma$ is {\em weakly acyclic} if the
dependency graph has no cycle going through a special edge.

 \waglav{} denotes the class  of all finite weakly acyclic sets of target tgds.
\end{definition}

 The tgd $\forall x\forall y(E(x,y) \rightarrow \exists z\, E(x,z))$ is  weakly
acyclic; in contrast,  the tgd   $\forall x \forall y(E(x,y) \rightarrow \exists z\,
E(y,z))$ is not, because the dependency graph contains a special
self-loop (see Figure~\ref{fig:dependency_graphs}).  Moreover, every set of GAV tgds is weakly acyclic, since the dependency graph contains no special edges in this case.

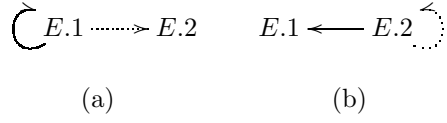
\begin{figure}
	\centering
	\begin{minipage}{0.4\columnwidth}
		\centerline{
		\xymatrix@=0.75em{ & & &\\
		& E.1 \ar@(dl,ul) \ar@{.>}[rr] & & E.2\\
		& & &\\
		}}
		\centerline{(a)}
	\end{minipage}
	\begin{minipage}{0.4\columnwidth}
		\centerline{
		\xymatrix@=0.75em{ & & &\\
		E.1 & & E.2 \ar[ll] \ar@{.>}@(dr,ur) & \\
		& & &\\
		}}
		\centerline{(b)}
	\end{minipage}
	\caption{The dependency graphs for (a) $\forall x\forall y(E(x,y) \rightarrow \exists z\, E(x,z))$ and (b) $\forall x \forall y(E(x,y) \rightarrow \exists z\, E(y,z))$.  Special edges are dotted.}
	\label{fig:dependency_graphs}
\end{figure}

What follows is the definition of the {\em chase procedure}.
\begin{definition}[chase procedure~\cite{DBLP:journals/tcs/FaginKMP05}]
	Let $K$ be an instance.
	\begin{description}
	
		\item[(tgd)] Let $d$ be a tgd $\phi({\bf x}) \to \exists {\bf y} \psi({\bf x},{\bf y})$.  Let $h$ be a homomorphism from $\phi({\bf x})$ to $K$ such that there is no extension of $h$ to a homomorphism $h'$ from $\phi({\bf x}) \wedge \psi({\bf x},{\bf y})$ to $K$.  We say that $d$ {\em can be applied to $K$ with homomorphism $h$}.
		
		Let $K'$ be the union of $K$ with the set of facts obtained by: (a) extending $h$ to $h'$ such that each variable in ${\bf y}$ is assigned a fresh labeled null, followed by (b) taking the image of the atoms of $\psi$ under $h'$.  We say that {\em the result of applying $d$ to $K$ with $h$} is $K'$, and write $K \xrightarrow{d,h} K'$.
		
		\item[(egd)] Let $d$ be an egd $\phi({\bf x}) \to (x_1=x_2)$.  Let $h$ be a homomorphism from $\phi({\bf x})$ to $K$ such that $h(x_1) \neq h(x_2)$.  We say that $d$ {\em can be applied to $K$ with homomorphism $h$}.  We distinguish two cases.
		\begin{itemize}
			\item If both $h(x_1)$ and $h(x_2)$ are in $Const$ then we say that {\em the result of applying $d$ to $K$ with $h$} is ``failure'', and write $K \xrightarrow{d,h} \bot$.
			\item Otherwise, let $K'$ be $K$ where we identify $h(x_1)$ and $h(x_2)$ as follows: if one is a constant, then the labeled null is replaced everywhere by the constant; if both are labeled nulls, then one is replaced everywhere by the other.  We say that {\em the result of applying $d$ to $K$ with $h$ is $K'$}, and write $K \xrightarrow{d,h} K'$.
		\end{itemize}
	\end{description}
	
	In the above, $K \xrightarrow{d,h} K'$ (including the case where $K'$ is $\bot$) is called a {\em chase step}.  We now define chase sequences and finite chases.
	
	Let $\Sigma$ be a set of tgds and egds, and let $K$ be an instance.
	\begin{itemize}
		\item A {\em chase sequence of $K$ with $\Sigma$} is a sequence (finite or infinite) of chase steps $K_i \xrightarrow{d_i,h_i} K_{i+1}$, with $i=0,1,...$, with $K = K_0$ and $d_i$ a dependency in $\Sigma$.
		\item A {\em finite chase of $K$ with $\Sigma$} is a finite chase sequence $K_i \xrightarrow{d,h} K_{i+1}$, $0 \le i < m$, with the requirement that either (a) $K_m = \bot$ or (b) there is no dependency $d_i$ of $\Sigma$ and there is no homomorphism $h_i$ such that $d_i$ can be applied to $K_m$ with $h_i$.  We say that $K_m$ is the result of the finite chase.  We refer to case (a) as the case of a {\em failing finite chase} and we refer to case (b) as the case of a {\em successful finite chase}.
	\end{itemize}
\end{definition}

In the context of data exchange, we chase the source instance first with the source-to-target constraints, and then continue chasing with the target constraints.  The nature of s-t tgds ensure that no atoms are created over the source schema, so in this setting the result of chasing a source instance $I$ with a schema mapping \sm{M} is a pair $(I,J)$ where $J$ is a target instance.  We usually refer to $J$ alone as the result of the chase.

We will also make use of the notion of \emph{rank}~\cite{DBLP:journals/tcs/FaginKMP05}. Let $\Sigma$ be a finite weakly acyclic set of tgds. For every node $(R, A)$ in the dependency graph of $\Sigma$,
define an {\it incoming path} to be any (finite or
infinite) path ending in $(R, A)$. Define the {\it rank} of $(R, A)$, denoted by
$rank{(R, A)}$, as the maximum number of special edges on any such
incoming path. Since $\Sigma$ is
weakly acyclic, there are no cycles going through special edges; hence,
$rank{(R, A)}$ is finite. The \emph{rank} of $\Sigma$, denoted $rank(\Sigma)$ is the maximum of $rank(R,A)$ over all positions $(R, A)$
in the dependency graph of $\Sigma$.

If $q$ is a query over the target schema $\bf T$ and $I$ is a source instance, then
the \emph{certain answers} of  $q$ with respect to $\mathcal {M}$ are defined as
	\begin{multline*}
		\certain(q,I,\sm{M})=\\
		\bigcap \{ q(J) :~ J \mbox{ is a solution for } I \mbox{ w.r.t. } \sm{M} \}
	\end{multline*}
	
\begin{definition}
	Let $J$ be an instance which may contain null values, and let $q$ be a conjunctive query over the schema of $J$.  Then $q\!\!\downarrow(J)$ is defined as the answers of $q$ on $J$ that contain no null values.
\end{definition}
If $J$ is a universal solution for a source instance $I$ w.r.t. a schema mapping \sm{M}, then for every conjunctive query $q$, it holds that $\certain(q,I,\sm{M}) = q\!\!\downarrow(J)$.

\subsection{Repairs and Consistent Answers.} Let $\Sigma$ be a set of constraints over some relational schema. An \emph{inconsistent} database is a database that violates at least one constraint in $\Sigma$. Informally, a \emph{repair} of an inconsistent database $I$ is a consistent database $I'$ that differs from $I$ in a ``minimal" way. This notion can be formalized in several different ways~\cite{DBLP:conf/pods/ArenasBC99}.
		\begin{enumerate}
			\item A \emph{symmetric-difference-repair} of $I$, denoted $\oplus$-repair of $I$, is an instance $I'$ that satisfies $\Sigma$ and where there is no instance $I''$ such that $ I \oplus I'' \subset I \oplus I' $ and $I''$ satisfies $\Sigma$. Here, $I\oplus I'$ denotes the set of facts that form the symmetric difference of the instances $I$ and $I'$.
			\item A \emph{subset-repair} of $I$ is an instance $I'$ that satisfies $\Sigma$ and where there is no instance $I''$ such that $ I' \subset I'' \subseteq I $ and $I''$ satisfies $\Sigma$.
			\item A \emph{superset-repair} of $I$ is an instance $I'$ that satisfies $\Sigma$ and where there is no instance $I''$ such that $ I' \supset I'' \supseteq I $ and $I''$ satisfies $\Sigma$.
		\end{enumerate}

Clearly, subset-repair and superset-repairs are also  $\oplus$-repairs; however, a $\oplus$-repair need not be a subset-repair or a superset-repair.

The \emph{consistent answers} of  a query $q$ on $I$ with respect to $\Sigma$ are defined as:
	\begin{multline*}
		\oplus\text{-}\cqa(q,I,\Sigma) = \\
		\bigcap \{ q(I') :~ I' \mbox{ is a } \mbox{$\oplus$-repair of } I \mbox{ w.r.t. }\Sigma \}
	\end{multline*}
	with subset and superset versions defined analogously.

\section{Framework and Related Work} \label{sec:xr}
In this section, we introduce the exchange-repair framework, discuss its structural and algorithmic properties, and explore its relationship to inconsistency tolerant semantics in data integration and ontology-based data access.

\subsection{The Exchange-Repair Framework}

	\begin{definition}
	Let  \mapping{} be a schema mapping, $I$ a source instance, and $(I',J')$ a pair of a source instance and a target instance.
		\begin{enumerate}
			\item We say that $(I',J')$ is a \emph{symmetric-difference exchange-repair solution} (in short, a
 $\oplus$-\xrsolution) for $I$ w.r.t.\ $\sm{M}$ if $(I',J')$ satisfies $\sm{M}$, and there is no pair of instances $(I'',J'')$ such that $ I \oplus I'' \subset I \oplus I' $ and $(I'',J'')$ satisfies $\sm{M}$.
			\item We say that $(I',J')$ is a \emph{subset exchange-repair solution} (in short, a subset-\xrsolution) for $I$ with respect to $\sm{M}$ if $I' \subseteq I$ and $(I',J')$ satisfies $\sm{M}$; and there is no pair of instances $(I'',J'')$ such that $ I' \subset I'' \subseteq I $ and $(I'',J'')$ satisfies $\sm{M}$.
		\end{enumerate}
		Note that  the minimality condition in the preceding definitions
applies to the source instance $I'$, but not to the target instance  $J'$ of the pair $(I',J')$.
 The source instance $I'$ of a $\oplus$-\xrsolution{} (subset-\xrsolution) for $I$ is called a {$\oplus$-\em source-repair} (respectively, subset source-repair) of $I$.
	\end{definition}

	\begin{figure*}[htbp]
		\Small\centering
		\begin{tabular}{rlcrl}
			\multicolumn{2}{c}{$(I,J)'$} & & \multicolumn{2}{c}{$(I,J)''$}\\
			\cline{1-2} \cline{4-5}
			\\
			\multicolumn{2}{c}{
			\begin{tabular}{|c|c|c|}
				\hline
				\multicolumn{3}{|c|}{\tt Task\_Assignments} \\
				person & task & dept \\
				\hline
				\sout{peter} & \sout{tpsreport} & \sout{software} \\
				\sout{peter} & \sout{spaceout} & \sout{software} \\
				peter & meetbobs & exec \\
				\hline
			\end{tabular}} & &
			\multicolumn{2}{c}{
			\begin{tabular}{|c|c|c|}
				\hline
				\multicolumn{3}{|c|}{\tt Task\_Assignments} \\
				person & task & dept \\
				\hline
				peter & tpsreport & software \\
				peter & spaceout & software \\
				\sout{peter} & \sout{meetbobs} & \sout{exec} \\
				\hline
			\end{tabular}}
			\\
			\\
			\begin{tabular}{|c|c|}
				\hline
				\multicolumn{2}{|c|}{\tt Departments} \\
				person & dept \\
				\hline
				~ & \\
				peter & exec \\
				~ & \\
				\hline
			\end{tabular} &
			\begin{tabular}{|c|c|}
				\hline
				\multicolumn{2}{|c|}{\tt Tasks} \\
				person & task \\
				\hline
				~ & \\
				~ & \\
				peter & meetbobs \\
				\hline
			\end{tabular} & &
			\begin{tabular}{|c|c|}
				\hline
				\multicolumn{2}{|c|}{\tt Departments} \\
				person & dept \\
				\hline
				peter & software \\
				~ & \\
				~ & \\
				\hline
			\end{tabular} &
			\begin{tabular}{|c|c|}
				\hline
				\multicolumn{2}{|c|}{\tt Tasks} \\
				person & task \\
				\hline
				peter & tpsreport \\
				peter & spaceout \\
				~ & \\
				\hline
			\end{tabular}
		\end{tabular}
		\caption{Two {\xrsolution}s for $I$ w.r.t. \sm{M}.  The {\tt Stakeholders} tables are omitted for brevity.}
		\label{example:exchange_repair}
	\end{figure*}

	Figure~\ref{example:exchange_repair} shows all two {\xrsolution}s for our source instance and schema mapping.  Notice that the shared origins of tuples are taken into account (for example, peter performs tasks only for his assigned department, unlike in Figure~\ref{fig:exchange_then_repair}), but the {\xrsolution}s retain more derived target information than the instances in Figure~\ref{fig:exchange_as_repair} (by preferring to satisfy tgds by adding rather than deleting tuples).  If we now evaluate ${\tt boss}({\rm peter},b)$ over each target instance, and take the intersection, we have $\{({\rm peter},{\rm bobs})\}$, which aligns well with our intuitive expectations.  A precise semantics for query answering is given later in this section.

 Source-repairs constitute a new notion that, in general, has different properties from those of the standard database repairs. Indeed, as mentioned earlier, a $\oplus$-repair need not be a subset repair. In contrast, Theorem~\ref{XR-subset-thm} (below) asserts that the state of affairs is different for source-repairs. Recall that, according to the notation introduced earlier, \mappingtypethree{\glav}{\waglav}{\egd} denotes the collection of all schema mappings \mapping{} such that \sigst{} is a finite set of s-t tgds and \sigt{} is the union of a finite weakly acyclic set of target tgds with a finite set of target egds.

	\begin{lemma} \label{subset-trickledown-lemma}
		Let \sm{M} be a \mappingtypethree{\glav}{\glav}{\egd} schema mapping.  If $I' \supseteq I$ are two source instances, then every solution for $I'$ w.r.t. \sm{M} is also a solution for $I$ w.r.t. \sm{M}, and consequently if $I$ has no solution w.r.t. \sm{M} then $I'$ has no solution w.r.t. \sm{M}.
	\end{lemma}
	
	\begin{proof}
		Let $I' \supseteq I$ be two source instances.  We will show that if $I'$ has a solution w.r.t. \sm{M} then $I$ also has a solution w.r.t. \sm{M}.  Let $J$ be an arbitrary solution for $I'$ w.r.t. \sm{M}.  Let $\phi({\bf x}) \to \exists {\bf y} \psi({\bf x},{\bf y})$ be an arbitrary tgd in \sigst{}, and let $h:{\bf x} \to {\rm adom}(I)$ be a homomorphism such that $h(\phi({\bf x})) \subseteq I$, and of course $h(\phi({\bf x})) \subseteq I'$ as well.  Then $h$ can be extended to some homomorphism $h'$ such that $h'(\psi({\bf x},{\bf y})) \subseteq J$, and therefore $(I,J)$ together satisfy \sigst{}, and since $J$ satisfies \sigt{}, we have that $J$ is also a solution for $I$ w.r.t. \sm{M}.
	\end{proof}

	\begin{theorem} \label{XR-subset-thm} Let ${\cal M}$  be a \mappingtypethree{\glav}{\glav}{\egd} schema mapping.  Let $I$ be a source instance.  Then if $(I',J')$ is a $\oplus$-\xrsolution{} of $I$ w.r.t.\ ${\cal M}$, then $(I',J')$ is actually a subset-\xrsolution{} of $I$ w.r.t.\ ${\cal M}$. Consequently, every $\oplus$-source-repair of $I$ is also a subset-source-repair of $I$.
	\end{theorem}

	\begin{proof}
		Let $(I',J')$ be an $\oplus$-\xrsolution{} for $I$ w.r.t. \sm{M}.  Suppose $I' \setminus I \neq \emptyset$.  Then by Lemma~\ref{subset-trickledown-lemma}, $J'$ is also a solution for $I' \cap I$.  Since $I' \oplus I \supset (I' \cap I) \oplus I$, we have that $(I',J')$ fails the minimality criterion and thus is not a $\oplus$-\xrsolution{} for $I$ w.r.t. \sm{M}, which is a contradiction.
	\end{proof}

\begin{remark}\em
  From here on and in view of Theorem \ref{XR-subset-thm}, we will use the term \xrsolution{} to mean subset-\xrsolution{}; similarly, source-repair will mean subset source-repair.
\end{remark}

Note that if ${\cal M}$ is a \mappingtypethree{\glav}{\waglav}{\egd} schema mapping, then source-repairs always exist. The reason is that, since the pair $(\emptyset,\emptyset)$ trivially satisfies ${\cal M}$,  then for every source instance $I$, there must exist a maximal subinstance $I'$ of $I$ for which a solution $J'$ w.r.t.\ ${\cal M}$ exists; hence, $(I',J')$ is a source repair for $I$ w.r.t.\ ${\cal  M}$.

We now claim that the following statements are true for arbitrary source instances and schema mappings.
\begin{enumerate}
	\item Repairs of the target instance obtained by chasing with the tgds of the schema mapping are not necessarily {\xrsolution}s.
	\item Repairs of $(I,\emptyset)$ are not necessarily {\xrsolution}s.
\end{enumerate}
For the first statement, consider the pair $(I,J')$ from Figures~\ref{fig:inconsistent_chase}~and \ref{fig:exchange_then_repair}, where $J'$ is a $\oplus$-repair of the inconsistent result $J$ of the chase of $I$. Clearly, $(I,J')$ is not an \xrsolution, because $J'$ is not a solution for $I$. For the second statement, consider the pairs $(I_1,J_1)$, $(I_2,J_2)$, $(I_3,J_3)$ in Figure \ref{fig:exchange_as_repair}, all of which are $\oplus$-repairs of  $(I,\emptyset)$. The first two are also {\xrsolution}s of $I$, but the third one is not.

It can also be shown that {\xrsolution}s are not necessarily $\oplus$-repairs of $(I,\emptyset)$.  We now describe an important case in which {\xrsolution}s {\bf are} $\oplus$-repairs of $(I,\emptyset)$. For this, we recall the notion of a \emph{core universal solution} from \cite{DBLP:journals/tcs/FaginKMP05}. By definition, a core universal solution is a universal solution that has no homomorphism to a proper subinstance. If a universal solution exists, then a core universal solution also exists. Moreover, core universal solutions are unique up to isomorphism.

\begin{proposition}
		\label{prp:core_xr_sol_is_repair}
Let ${\cal M}$  be a \mappingtypethree{\glav}{\glav}{\egd} schema mapping. If $I$ is source instance and	
	$(I',J')$  is an \xrsolution{} for $I$ w.r.t. $\sm{M}$ such that $J'$ is a core universal solution for $I'$ w.r.t. $\sm{M}$, then $(I',J')$ is a $\oplus$-repair of $(I,\emptyset)$ w.r.t.\ $\sigst{} \cup \sigt{}$.
	\end{proposition}
	
\begin{proof}
	Let $(I'', J'')$ be a pair of instances which together satisfy $\sigst \cup \sigt$ and such that $(I'', J'') \oplus (I,\emptyset) \subseteq (I', J') \oplus (I,\emptyset)$.  Since $(I', J')$ is an \xrsolution{} for $I$ w.r.t. $\sm{M}$, there is no instance $I''$ such that $I' \subset I'' \subseteq I$ and $I''$ has a solution w.r.t. $\sm{M}$.  Therefore it must be that $I'' = I'$.  Furthermore, since $J'$ is a core universal solution, there is no proper subinstance $J'' \subset J'$ that is a solution for $I'$ w.r.t. $\sm{M}$, so $J'' = J'$.  Therefore $(I', J')$ is a $\oplus$-repair of $(I,\emptyset)$ w.r.t.\ $\sigst{} \cup \sigt{}$. 
\end{proof}

Next, we present the second key notion in the exchange-repair framework.

\begin{definition} Let \mapping{} be a schema mapping and $q$ a query over the target schema $\bf T$. If $I$ is a source instance, then the \emph{XR-certain answers} of $q$ on $I$ w.r.t.\ $\cal M$ is the set
		\begin{multline*}
			\xrcertain(q,I,\sm{M})=\\
			\bigcap \{ q(J'):\mbox{$(I',J')$ is an \xrsolution{} for $I$} \}.
		\end{multline*}
	\end{definition}
	
Note that when $I$ has a solution w.r.t. \cal{M}, it is its own only \xrsolution.  Thus the \xrcertain{} semantics coincide with certain semantics when solutions exist.  The next results provide a comparison of the XR-certain answers with the consistent answers.

	\begin{proposition} \label{xr-certain-vs-cqa-prop}
Let \mapping{} be a \mappingtypethree{\glav}{\waglav}{\egd} schema mapping and $q$ a conjunctive query over the target schema $\bf T$. If $I$ is a source instance, then $\xrcertain(q,I,\sm{M}) \supseteq \oplus\text{-}\cqa(q,(I, \emptyset),\sigst{}\cup\sigt{})$.  Moreover, this containment may be a proper one.
	\end{proposition}

	\begin{proof}
		Since \sm{M} is weakly acyclic, for any instance $I$ for which solutions exist, a core universal solution also exists.  Therefore, we have that $\xrcertain(q,I,\sm{M})=\bigcap \{ q(J'):(I',J')$ is an \xrsolution{} for $I$ w.r.t. \sm{M}, and $J'$ is a core universal solution for $I'$ w.r.t. \sm{M} $\}$.
		By Proposition~\ref{prp:core_xr_sol_is_repair}, the set of {\xrsolution}s $(I',J')$ where $J'$ is a core universal solution for $I'$ w.r.t. \sm{M} is a subset (maybe proper) of the set of $\oplus$-repairs of $(I,\emptyset)$ w.r.t. $\sigst \cup \sigt$.  Therefore $\xrcertain(q,I,\sm{M}) \supseteq \oplus\text{-}\cqa(q,(I, \emptyset),\sigst{}\cup\sigt{})$.
		
		To see that this containment may be a proper one, consider the schema mapping \sm{M} and query ${\tt boss}({\tt peter},b)$ in Figure~\ref{example:mapping}, and the repairs of $(I,\emptyset)$ in Figure~\ref{fig:exchange_as_repair}.  It is easy to verify that $\oplus\text{-}\cqa({\tt boss}({\tt peter},b),(I, \emptyset),\sigst{}\cup\sigt{})=\emptyset$, while $\xrcertain({\tt boss}({\tt peter},b),I,\sm{M}) =\{({\tt peter},{\tt bobs})\}$.  
	\end{proof}
	
	The following proposition pertains to the case where \sigst{} is the {\em copy mapping}, i.e. for each relation $R \in \bf S$ there is a corresponding relation $R'$ of the same arity in $\bf T$, and \sigst{} contains only the tgd $R({\bf x}) \to R'({\bf x})$ for each $R \in \bf S$.    We say an instance $J$ is the {\em copy} of an instance $I$ if $J$ is the canonical universal solution for $I$ w.r.t. the copy mapping (so it contains the same facts up to renaming of relations).
	\begin{proposition} \label{prp:xr-certain-copy-is-cqa}
		Let \mapping{} be a \mappingtypetwo{\gav}{\egd} schema mapping where \sigst{} is the copy mapping, and let $q$ be a conjunctive query over the target schema $\bf T$.  Then for every instance $I$, it holds that $\xrcertain(q,I,\sm{M}) = \text{subset-}\cqa(q,J,\sigt)$, where $J$ is the copy of $I$.
	\end{proposition}

	\begin{proof}
		Since \sigst{} specifies the copy mapping and \sigt{} contains only egds, for every source repair $I'$ there is an \xrsolution $(I',J')$ where $J'$ is the copy of $I'$.  Furthermore, $J'$ is a universal solution for $I'$ w.r.t. \sm{M}, so we can write $\xrcertain(q,I,\sm{M}) = \bigcap \{ q(J') ~|~ (I',J')$ is an \xrsolution{} for $I$ w.r.t. \sm{M} and $J'$ is the copy of $I' \}$.  Therefore $\xrcertain(q,I,\sm{M}) = \bigcap \{ q(J') ~|~ J' $ is the copy of a maximal subset of $I$ such that $J' \models \sigt\}$.  Let $J$ be the copy of $I$.  Then $\xrcertain(q,I,\sm{M}) = \bigcap \{ q(J') ~|~ J'$ is a maximal subset of $J$ such that $J' \models \sigt \}$, which is precisely \text{subset-}$\cqa(q,J,\sigt)$. 
	\end{proof}

Let  \mapping{} be a schema mapping and $q$ a Boolean query over  $\bf T$.
We consider two natural decision problems in the exchange-repair framework, and give upper bounds for their computational complexity.

\smallskip

\begin{itemize}

\item \textbf{Source-Repair Checking}:
Given a source instance $I$ and a source instance $I'\subseteq I$, is $I'$ a source-repair of $I$ w.r.t.\ ${\cal M}$?

\item \textbf{\xrcertain{} Query Answering}:
Given a source instance $I$, does $\xrcertain(q,I,\sm{M})$ evaluate to true? In other words, is $q(J')$  true on every target instance $J'$ for which there is a source instance $I'$ such that $(I',J')$ is an \xrsolution{} for $I$?
\end{itemize}
	
\begin{theorem} \label{xr-certain-compl-thm}
Let ${\cal M}$ be a  \mappingtypethree{\glav}{\waglav}{\egd} schema mapping.
\begin{enumerate}
\item The source-repair checking problem is in \ptime.
\item Let $q$ be a union of conjunctive queries over the target schema. The \xrcertain{} query answering problem for $q$ is in \coNP.
    \end{enumerate}
    Moreover, there is a schema mapping specified by copy s-t tgds and target egds, and a Boolean conjunctive query  for which the \xrcertain{} query answering problem is \coNP-complete. Thus, the data complexity of the \xrcertain{} answers for Boolean conjunctive queries is \coNP-complete.
    \end{theorem}

    \begin{proof}
    For the first part, the following is a polynomial time algorithm to check if $I' \subseteq I$ is a source repair of $I$ w.r.t. \sm{M}:
    	\begin{quote}
    		Use the chase procedure to check that $I'$ has a solution w.r.t. \sm{M}~\cite{DBLP:journals/tcs/FaginKMP05}.  For every tuple $t \in I \setminus I'$, use the chase procedure to check that $I' \cup \{ t \}$ does {\em not} have a solution w.r.t. \sm{M}.
    	\end{quote}
    The first step ensures that $I'$ has a solution, and by Lemma~\ref{subset-trickledown-lemma}, the second step is sufficient to ensure that $I'$ is a maximal such subset of $I$.  Since \sm{M} is weakly acyclic, this algorithm runs in time which is polynomial in the size of $I$.
    
	For the second part, the following is an algorithm in \NP to check if $\xrcertain(q,I,\sm{M})$ is false:
		\begin{quote}
			Let $I'$ be an arbitrary subset of $I$.  Using the algorithm from the first part, check that $I'$ is a source repair of $I$.  If so, check that $q(chase(I')) = false$.
		\end{quote}
	For the matching lower bound, consider the schema mapping \mapping{} and target conjunctive query $q$, where ${\bf S} =\{\rel P(x,y), \rel Q(x,y)\}$, ${\bf T} = \{\rel P'(x,y),\rel Q'(x,y)\}$, $\sigst = \{ \rel P(x,y) \to \rel P'(x,y), \rel Q(x,y) \to \rel Q'(x,y) \}$, and $\sigt = \{ \rel P'(x,y) \wedge \rel P'(x,y') \to y=y', \rel Q'(x,y) \wedge \rel Q'(x,y') \to y=y' \}$, and $q = \exists x \exists y \exists x' \rel P'(x,y) \wedge \rel Q'(x',y)$.  Note that \sigst{} is the {\em copy} mapping, therefore, we have $\xrcertain(q,I,\sm{M}) = \oplus\text{-}\cqa(q,J,\sigt)$ where $J$ is merely a copy of $I$.  For the given target query and target constraints, the latter is known to be \coNP-hard in data complexity \cite{DBLP:journals/iandc/ChomickiM05,DBLP:journals/jcss/FuxmanM07}. 
    \end{proof}

Theorem \ref{xr-certain-compl-thm} implies that the algorithmic properties of exchange-repairs are quite different from those of $\oplus$-repairs. Indeed, as shown in \cite{DBLP:conf/icdt/AfratiK09,DBLP:conf/icdt/CateFK12}, for
  \mappingtypethree{\glav}{\waglav}{\egd} schema mappings,  the $\oplus$-repair-checking problem is in \coNP (and can be \coNP-complete), and the data complexity of the consistent answers of Boolean conjunctive queries is
  $\Pi_2^p$-complete~\cite{DBLP:journals/iandc/ChomickiM05}.  This drop in complexity can be directly attributed to Theorem~\ref{XR-subset-thm}.

\subsection{Related Work}\label{sec:related}
The work reported here builds directly on the work of many others, in particular  the foundational work on database repairs and consistent query answering by Arenas, Bertossi, and Chomicki~\cite{DBLP:conf/pods/ArenasBC99}, and on data exchange and certain query answering by Fagin et al.~\cite{DBLP:journals/tcs/FaginKMP05}.

As mentioned earlier, the main conceptual contribution of this paper is the introduction of an inconsistency-tolerant semantics for data exchange, called exchange-repairs, in which we consider repairs to the source instance.
Inconsistency-tolerant semantics have been studied in several different areas of database management, including data integration and ontology-based data access (OBDA). The common motivation for inconsistency-tolerant semantics  is to give non-trivial and, in fact, meaningful semantics to query answering. We now discuss the relationship between the XR-certain answers and the inconsistency-tolerant semantics of queries in these different contexts.

\subsection{Connections with data integration} \label{data-integration-connections}


In \cite{DBLP:conf/ijcai/CaliLR03} and \cite{DBLP:conf/krdb/LemboLR02}, the authors introduce and study the notion of \emph{loosely-sound} semantics for queries in a data integration setting. There are two main differences between that setting and ours. To begin with, they consider schema mappings in which the schema mapping consists of GAV (global-as-view) constraints between the source (local)  schema and the target (global) schema, and also of key constrains and inclusion dependencies on the target schema; in contrast, we consider richer constraint languages, namely, GLAV (global-and-local-as-view) constraints between source and target, and also target egds and target tgds.  More importantly perhaps, the loosely-sound semantics are, in general, different from the XR-certain answers semantics. Specifically, given a source instance $I$, the loosely-sound semantics are obtained by first computing the result $J$ of the chase of $I$ with the GAV constraints between the source and the target, and then considering as ``repairs" all instances $J'$ that satisfy the target constraints and are inclusion maximal in their intersection with $J$. If all target constraints are egds (in particular, if all target constraints are key constraints), then it is easy to show that, for target conjunctive queries, the loosely-sound semantics  coincide with the consistent answers of queries with respect to subset repairs of $J$. Thus, in this case, the loosely-sound semantics give the same unsatisfactory answers as the materialize-then-repair approach seen in Figure~\ref{fig:exchange_then_repair}.
 Concretely, this approach yields the instance $J'$ in Figure~\ref{fig:exchange_then_repair} as one possible ``repair" of the instance $J$ in Figure 1, and includes the undesirable answers $({\rm peter},{\rm portman})$ and $({\rm peter},{\rm lumbergh})$ to the query ${\tt boss}({\rm peter},b)$.
 Thus, this same example shows that the loosely-sound semantics are different from the XR-certain semantics.

In~\cite{DBLP:conf/pods/CaliLR03}, Cal\`{\i}, Lembo, and Rosati consider the notions of \emph{loosely-sound}, \emph{loosely-complete}, and \emph{loosely-exact} semantics of queries on an inconsistent database. We note that the loosely-exact semantics coincide with the consistent-answer semantics with respect to symmetric-difference-repairs of the inconsistent database.

\subsection{Connections with ontology-based data access} \label{obda-connections}
Ontology-based data access (OBDA), originally introduced in \cite{CGLLPR07}, is a framework for answering queries over knowledge bases.  In that framework, a knowledge base over a schema $\bf T$ is a pair $(D,\Sigma)$, where $D$ is a $\bf T$-instance and $\Sigma$ is a set of constraints expressed in some logical formalism over $\bf T$. The instance $D$ represents extensional knowledge given by the facts of $D$, and is called the ABox. The set $\Sigma$ of constraints represents intensional knowledge, and is called the TBox.  In most scenarios, the schema $\bf T$ consists of unary relation symbols, called  \emph{concepts}, and of binary relation symbols, called  \emph{roles}. Moreover, $\Sigma$ typically consists  of sentences in some description logic.
An inconsistency-tolerant  semantics in the context of  OBDA was first investigated in
\cite{DBLP:conf/rr/LemboR07}; this semantics is based on the notion of \emph{AR-repairs} (ABox-repairs) and has become known as \emph{AR-semantics}.  Subsequent investigations of  AR-semantics were carried out in a number of papers,  including (in chronological order) \cite{LemboLRRS10,Rosati11,Bienvenu12,BienvenuR13,BienvenuBG14,LMPS15}. These papers have analyzed  the computational complexity of consistent query answering in OBDA and have also considered several variants of the AR-semantics in the OBDA framework.

Data exchange and OBDA are different frameworks that aim to formalize different aspects of data inter-operability.  Ιn data exchange there are two schemas, the source schema and the target schema, with no restrictions on the type of relation symbols they contain,  while in OBDA there is a single schema that typically contains only unary and binary relation symbols. Moreover, as seen in the preceding discussion, the constraints  typically used in data exchange are quite different from those typically used in OBDA.
One notable exception to this is the work reported in \cite{LMPS15}, where the OBDA framework studied allows for tuple-generating dependencies (it also allows for \emph{negative constraints}, but not for equality-generating dependencies).
In spite of these differences, it turns out that there are close connections between data exchange and OBDA.  In what follows, we spell out these connections in detail and show that, as regards consistent query answering, each of these two frameworks can simulate the other.

We first introduce some basic concepts and terminology for OBDA; for the most part, we  follow  \cite{LMPS15}.

Let $\bf T$ be a schema and let $(D,\Sigma)$ be a knowledge base over $\bf T$.
A \emph{model} of  $(D,\Sigma)$ is a $\bf T$-instance $J$ such that $D\subseteq J$ and $J \models \Sigma$.
    We write $\mds(D,\Sigma)$ for the set of all models of $(D,\Sigma)$.

 An \emph{AR-repair} of $(D,\Sigma)$  is a ${\bf T}$-instance  $D'$ with the following properties:
 (i)  $D'\subseteq   D$; (ii)   $\mds(D',\Sigma) \not = \emptyset$; (iii)
   $D'$ is an inclusion maximal sub-instance of $D$ having the second property, i.e., there is no $\bf T$-instance $D''$ such that $D' \subset D''\subseteq D$ and $\mds(D'',\Sigma) \not = \emptyset$.
We write $\drep(D,\Sigma)$ for the set of all AR-repairs of $(D,\Sigma)$.

Next, we introduce the notion of consistent query answering in the context of OBDA.
Let $q$ be a Boolean query over the schema $\bf T$. We say that $q$ is \emph{entailed by $(D,\Sigma)$ under AR-semantics}
 if for every AR-repair $D'$ in $\drep(D,\Sigma)$  and every
$\bf T$-instance $J\in \mds(D',\Sigma)$, we have that $J\models q$.
If $q$ is a non-Boolean query of arity $k$ over the schema $\bf T$ and $\bf a$ is a $k$-tuple of constants, then we say that $q({\bf a})$ is \emph{entailed by $(D,\Sigma)$ under AR-semantics}
if $q({\bf a})$ is entailed when viewed as a Boolean query; this  means that for every AR-repair $D'$ in $\drep(D,\Sigma)$  and every
$\bf T$-instance $J\in \mds(D',\Sigma)$, we have that   ${\bf a}$ belongs to the result $q(J)$ of evaluating $q$ on $J$.
We write
 $\ARcertain(q,D,\Sigma)$ to denote the set of all tuples $\bf a$ such that $q({\bf a})$ is entailed by $(D,\Sigma)$ under AR-semantics.
  By unraveling the definitions, we see that
  \begin{multline*}
  \ARcertain(q,D,\Sigma)= \\
   \bigcap\{q(J): \mbox{$J \in \mds(D',\Sigma)$ and
     $D' \in \drep(D,\Sigma)$}\}.
   \end{multline*}
We are now ready to establish the precise connections between the exchange-repairs framework and the OBDA framework.

\subsubsection{From OBDA to exchange repairs}
Assume that  $(D,\Sigma)$ is a knowledge base over a schema $\bf T$.
Let ${\bf S^*}$ be the schema of the relation symbols occurring in $D$; note that ${\bf S^*}$ is a (possibly proper) subschema of $\bf T$.  Let $\bf S$ be a \emph{copy} of ${\bf S^*}$, that is, for every relation symbol $R^*$ in ${\bf S^*}$, there is a relation symbol $R$ in $\bf S$ of the same arity.
If $K$ is an ${\bf S^*}$-instance, we will write $K_{\bf S}$ to denote $\bf S$-copy of $K$, i.e., the ${\bf S}$-instance obtained from $K$ by renaming the facts of $K$ using the corresponding relation symbols in $\bf S$. Conversely, if $I$ is an $\bf S$-instance, then we will write $I_{\bf S^*}$ to denote the ${\bf S^*}$-copy of $I$.

The next proposition tells that the OBDA framework can be simulated by the exchange-repairs framework.
The proof is  straightforward, and it is omitted.
\begin{proposition} \label{OBDA-to XR-prop}
Consider the schema mapping ${\mathcal M} = ({\bf S},{\bf T}, \sigst, \Sigma)$, where $\sigst$ is the set of
 copy s-t tgds from ${\bf S}$ to ${\bf S^*}$.
The following statements are true.
\begin{enumerate}
\item  For every ${\bf T}$-instance $D' \subseteq D$ and every ${\bf T}$-instance J, we have that
$D' \subseteq J$ and $J \models \Sigma$ if and only if $J$ is a solution for $D'_{\bf S}$ w.r.t.\ ${\mathcal M}$.
\item  For every ${\bf S}$-instance $I' \subseteq D_{\bf S}$  and every $\bf T$-instance $J$, we have that
$J$ is a solution for $I'$ w.r.t.\ ${\mathcal  M}$ if and only if $I'_{\bf S^*} \subseteq J$ and $J \models \Sigma$.

\item  There is a 1-1 correspondence between the AR-repairs of $(D,\Sigma)$ and the subset source repairs of $D_{\bf S}$ w.r.t.\ ${\mathcal  M}$.
     In fact, they are the same up to renaming relation symbols in ${\bf S^*}$ by their copies in ${\bf S}$.

\item   For every query $q$ over $\bf T$, we have that $\ARcertain(q,D,\Sigma)  =  \xrcertain(q,D_{\bf S},{\mathcal M})$.
\end{enumerate}
\end{proposition}

\subsubsection{From exchange repairs to OBDA}

Assume that ${\mathcal M} = ({\bf S},{\bf T}, \sigst, \sigt)$ is a schema mappings  in which $\sigst$ is a set of s-t tgds and $\sigt$ is a set of  arbitrary constraints over $\bf T$.  Recall that the schemas ${\bf S}$ and ${\bf T}$ have no relation symbols in common.
The next proposition tells that the exchange-repairs framework can be simulated by the OBDA framework.
Since $\Sigma_t$ are arbitrary constraints, Theorem~\ref{XR-subset-thm} does not necessarily apply, so we explicitly focus on {\em subset} source repairs.

\begin{proposition} \label{XR-to-OBDA-prop}
Let $I$ be a source instance. Consider the knowledge base $(I,\sigst \cup \sigt)$ with $I$ as the ABox and the union $\sigst\cup \sigt$ over the schema ${\bf S} \cup {\bf T}$ as the TBox. The following statements are true.
\begin{enumerate}
 \item For every source instance $I'$, we have that
$I'$ is a subset source repair of $I$ w.r.t.\ ${\mathcal M}$  if and only if $I'$ is an AR-repair of $(I, \sigst\cup \sigt)$.
\item For every query $q$ over $\bf T$, we have that
$\xrcertain(q,I,{\mathcal M})   = \ARcertain(q,I,  \sigst \cup \sigt)$.
\end{enumerate}
\end{proposition}
\begin{proof}

Assume first that $I'$ is  a subset source repair of $I$ w.r.t.\ ${\mathcal M}$. We have to show that $I'$ is an AR-repair of $(I, \sigst\cup \sigt)$.
 Since $I'$ is a subset source repair of $I$ w.r.t.\ $\mathcal M$, we have that $I'\subseteq I$. Moreover,
  there is a solution $J$ for $I'$ w.r.t.\ $\mathcal M$. Let $J' = I' \cup J$. Clearly, we have that (i) $I'\subseteq J'$ and (ii)
  $J' \models \sigst \cup \sigt$, hence $\mds(I', \sigst \cup \sigt) \not = \emptyset$. It remains to show that $I'$ is a maximal sub-instance of $I$ with the preceding properties (i) and (ii). Towards a contradiction, suppose that there is a sub-instance $I''$ of $I$ such that $I' \subset I'' \subseteq I$ and $\mds(I'', \sigst \cup \sigt) \not = \emptyset$. Consider an instance $J'' \in \mds(I'', \sigst \cup \sigt)$. Then $I''\subseteq J''$ and $J''\models
  \sigst \cup \sigt$. Let $\restr{J''}{\bf T}$ be the restriction of $J''$ to the target schema ${\bf T}$, that is,  $\restr{J''}{\bf T}$ is the sub-instance of $J''$ consisting of the facts of $J''$ that
   involve  relation symbols in $\bf T$ only.
  We claim that $\restr{J''}{\bf T}$ is a solution for $I''$ w.r.t.\ $\mathcal M$. Indeed, $\restr{J''}{\bf T} \models \sigt$, since all formulas in $\sigt$ contain atomic formulas from $\bf T$ only. Moreover, since $I''\subseteq J''$ and since $J'' \models \sigst$, we have $(I'',\restr{J''}{\bf T})\models \sigst$. This is so because, since  $\sigst$ consists of s-t tgds, the ${\bf S}$-facts in $J''\setminus I''$ play no role in  satisfying $\sigst$; we note that this may not hold if, say, $\sigst$ contained target-to-source tgds.
  It follows that $I'$ is not a subset source repair for $I$ w.r.t.\ $\mathcal M$, which is a contradiction.

Next, assume that $I'$ is an AR-repair of $(I,\sigst \cup \sigt)$. We have to show that $I'$ is a subset source repair of $I$ w.r.t.\ $\mathcal M$. Since $I'$ is an AR-repair of $(I,\sigst \cup \sigt)$, we have that
$I'\subseteq I$ and $\mds(I',\sigst \cup \sigt) \not = \emptyset$. Let $J'$ be a member of $\mds(I',\sigst \cup \sigt)$. Hence, $I' \subseteq J'$ and  $J'\models \sigst \cup \sigt$. If $\restr{J'}{\bf T}$ is the restriction of $J'$ to the target schema $\bf T$, then  $(I',\restr{J'}{\bf T}) \models \sigst$ (because $\sigst$ consists of s-t tgds) and $\restr{J'}{\bf T}\models \sigt$. Thus, there is a solution for $I'$ w.r.t.\ $\mathcal M$. Moreover, we claim that $I'$ is a maximal sub-instance of $I$ for which there exists a solution w.r.t.\ $\mathcal M$. Indeed, if $I''$ is  such that $I' \subset I''\subseteq I$ and
a solution $J''$ for $I''$ w.r.t.\ $\mathcal M$ exists, then $I''\cup J'' \models \sigst \cup \sigt$. It follows that $I'$ is an AR-repair of $(I,\sigst \cup \sigt)$, which is a contradiction.

Finally, if  $q$ is a query over $\bf T$, then, using the first part of the proposition, it is easy to show that
$\xrcertain(q,I,{\mathcal M})   = \ARcertain(q,I,  \sigst \cup \sigt)$.
\end{proof}

\section{\cqa{}-Rewritability}\label{sec:CQA-rewritability}

In this section, we show that, for \mappingtypetwo{\gav}{\egd} schema
mappings $\mapping{}$, it is possible to construct a set of egds
$\Sigma_s$ over $\mathbf{S}$ such that an $\mathbf{S}$-instance $I$ is
consistent with $\Sigma_s$ if and only if $I$ has a solution
w.r.t. $\mathcal{M}$. We use this to show that
$\xrcertain(q,I,\sm{M})$ for a conjunctive query $q$ coincides with
$\text{subset-\cqa}(q_s,I,\Sigma_s)$ for a union of conjunctive queries $q_s$. Thus, we can employ tools for
consistent query answering with respect to egds in order to compute
XR-certain answers for \mappingtypetwo{\gav}{\egd} schema
mappings.

We will use the well-known technique of \emph{GAV unfolding}
(see, e.g., \cite{DBLP:conf/pods/Lenzerini02}).
            Let $\mapping$ be a \mappingtypetwo{\gav}{\egd} schema
mapping.	For each $k$-ary target relation $T\in\mathbf{T}$,
                let $q_{\textsc{t}}$ be the set of all conjunctive queries
                $q(x_1, \ldots, x_k) = \exists\mathbf{y} (\phi(\mathbf{y}) \wedge x_1=y_{i_1}\wedge \cdots \wedge x_k=y_{i_k})$, for 
                $\phi(\mathbf{y}) \to T(y_{i_1}, \ldots, y_{i_k})$ a GAV tgd belonging to
                $\Sigma_{st}$ (recall that we frequently omit universal
                quantifiers in our notation, for the sake of readability).

		A {\em GAV unfolding} of a conjunctive query
                $q(\mathbf{z})$ over $\mathbf{T}$ w.r.t. $\Sigma_{st}$ is a
                conjunctive query over $\mathbf{S}$ obtained
               by replacing each occurrence of a
                target atom $T(\mathbf{z'})$ in $q(\mathbf{z})$ with one of the conjunctive
                queries in $q_{\textsc{t}}$ (substituting variables from $\mathbf{z'}$ for $x_1,\ldots,x_k$, and pulling existential quantifiers out to the front of the formula).

	        Similarly, we define a {\em GAV unfolding} of an egd $\phi(\mathbf{x}) \to
                x_k=x_l$ over $\mathbf{T}$ w.r.t. $\Sigma_{st}$ to be an egd over
                $\mathbf{S}$ obtained by 
 replacing each occurrence of a
                target atom $T(\mathbf{z'})$ in $\phi(\mathbf{x})$
by one of the conjunctive
                queries in $q_{\textsc{t}}$ (substituting variables from $\mathbf{z'}$ for $x_1,\ldots,x_k$, and pulling existential quantifiers
                out to the front of the formula as needed, where they
                become universal quantifiers).
                
	Figure~\ref{fig:gav_unfolding} shows the GAV unfolding of the schema mapping and query from Figure~\ref{example:mapping}.  
    
\begin{figure*}
\centering\normalsize
	\fbox{\parbox{0.9\textwidth}{
	$\sigs = \left\{ \begin{array}{l} {\tt Task\_Assignments}(p,t,d) \wedge {\tt Task\_Assignments}(p,t',d') \to d=d' \end{array} \right\}$
	\vspace{-2mm}
	\begin{multline*}
		{\tt boss}_s(person,stakeholder) = \exists task,department (\\
		{\tt Task\_Assignments}(person, task, department) \wedge {\tt Stakeholders\_old}(task,stakeholder))
	\end{multline*}\vspace{-5mm}}}
	\caption{The GAV Unfolding of the schema mapping and query given in Figure~\ref{example:mapping}.}
	\label{fig:gav_unfolding}
\end{figure*}

	\begin{theorem} \label{thm:xr-to-cqa}
		Let \mapping{} be a \mappingtypetwo{\gav}{\egd} schema
                mapping, and let $\Sigma_s$ be the set of all GAV
                unfoldings of egds in \sigt{}
                w.r.t. \sigst{}.  Let $I$ be an $\mathbf{S}$-instance.  The the following are equivalent:
               \begin{enumerate}
                 \item $I$ satisfies
                   $\Sigma_s$ if and only if $I$ has a
                   solution w.r.t. $\sm{M}$.
                 \item The subset-repairs of $I$
                   w.r.t. $\Sigma_s$ are the source repairs of
                   $I$
                   w.r.t. $\sm{M}$.
                 \item For each conjunctive query $q$ over
                   $\mathbf{T}$, we have that $\xrcertain(q,I,\sm{M}) =
                   \text{subset-}\cqa(q_s,I,\sigs)$, where $q_s$
                   is the union of GAV-unfoldings of $q$ w.r.t. \sigst{}.
                 \end{enumerate}
	\end{theorem}

	\begin{proof}
	
	\begin{enumerate}
		\item Let $I$ be an $\bf S$-instance which does not satisfy $\Sigma_s$.  Then there is an egd $\phi_1({\bf x}) \wedge ... \wedge \phi_k({\bf x}) \to x_i = x_j \in \sigs$ which is violated in $I$ by some image $\phi_1({\bf a}) \wedge ... \wedge \phi_k({\bf a})$.  By the definition of \sigs{}, there is an egd $T_1({\bf x}) \wedge ... \wedge T_k({\bf x}) \to x_i = x_j$ in \sigt{}, and tgds $\phi_1({\bf x}) \to T_1({\bf x}), ..., \phi_k({\bf x}) \to T_k({\bf x})$ in \sigst{}.  Then for any instance $J$ where $(I,J)$ together satisfy \sigst{}, it holds that $J$ contains the image $T_1({\bf a}) \wedge ... \wedge T_k({\bf a})$ and therefore violates \sigt{}.  The proof of the converse is similar.
		\item Consider that the source repairs are the maximal subsets of $I$ for which solutions exist.  Using the above, we have that these are also the maximal subsets of $I$ which satisfy \sigs{}, and therefore they are also the subset repairs of $I$ w.r.t. \sigs{}.
		\item By definition $\xrcertain(q,I,\sm{M})$ is the intersection over $q(J')$ for all \xrsolution{}s $(I',J')$ w.r.t. \sm{M} (or in other words, for all source repairs $I'$ and solutions $J'$ for $I'$ w.r.t. \sm{M}).  Observe that this is the intersection of $\certain(q,I',\sm{M})$ over all source repairs $I'$ w.r.t. \sm{M}.  We will now show that $\certain(q,I',\sm{M}) = q_s(I')$:
			\begin{quote} Let $J'$ be the solution for $I'$ w.r.t. \sm{M}.  Suppose $\bf a$ is a tuple in $q(J')$.  Then there is some image $T_1({\bf a},{\bf b}) \wedge ... \wedge T_k({\bf a},{\bf b})$ of $q$ in $J'$, and there are some tgds $\phi_1({\bf x},{\bf y}) \to T_1({\bf x},{\bf y}), ..., \phi_k({\bf x},{\bf y}) \to T_k({\bf x},{\bf y})$ in \sigst{} where the image $\phi_1({\bf a},{\bf b}) \wedge ... \wedge \phi_k({\bf a},{\bf b})$ is in $I'$.  By definition the clause $\exists y \phi_1({\bf x},{\bf y}) \wedge ... \wedge \phi_k({\bf x},{\bf y})$ is in $q_s$, so $\bf a$ is in $q_s(I')$.  The proof of the converse is similar.
			\end{quote}
		We now have that $\xrcertain(q,I,\sm{M})$ is the intersection over $q_s(I')$ for all source repairs $I'$ of $I$ w.r.t. \sm{M}.  By the second item of the theorem, this gives the intersection over $q_s(I')$ for all subset repairs $I'$ of $I$ w.r.t. \sigs{}, which is simply subset-$\cqa(q_s,I,\sigs)$.  
	\end{enumerate}
	\end{proof}

The following result tells us that Theorem~\ref{thm:xr-to-cqa}
cannot be extended to schema mappings containing LAV s-t tgds. 
	\begin{theorem}\label{thm:reachability-counterexample}
		Consider the \mappingtypetwo{\lav}{\egd} schema mapping \mapping{}, where
                \begin{itemize}
                  \item		$\mathbf{S}=\{R\}$ and
                    $\mathbf{T}=\{T\}$,
                  \item $\Sigma_{st} = \{R(x,y) \to \exists u~T(x,u)\land T(y,u)\}$, and
                    \item $\Sigma_t = \{T(x,y)\land T(x,z)\to y=z\}$. 
                \end{itemize}
                Consider the query
		$q(x,y) = \exists z.~T(x,z)\land T(y,z)$ over $\mathbf{T}$.
		There does not exist a UCQ $q_s$ over $\mathbf{S}$ and
                a set
                of universal first-order sentences (in particular, egds)
		$\Sigma_s$ such that, for every instance $I$, we have that
		$\xrcertain(q,I,\sm{M}) = \text{subset-}\cqa(q_s,I,\Sigma_s)$.
	\end{theorem}

	It is worth noting that the schema mapping $\sm{M}$ in the
        statement of Theorem~\ref{thm:reachability-counterexample}
         is such that every source instance has a
        solution, and hence ``$\xrcertain$'' could be replaced by
        ``$\certain$'' in the statement.

	\begin{proof}
		We start by observing that $\certain(q,I,\sm{M})$ expresses undirected reachability
		along the relation $R$:
		\begin{trivlist}
			\item \emph{Claim:} For every $\textbf{S}$-instance $I$,
			 $\certain(q,I,\sm{M}) = \{(a,b)\in adom(I) \mid b$ is reachable from $a$
			 by an undirected $R$-path$\}$.
		\end{trivlist}
		The left-to-right inclusion can be proved by induction on the length of the
		shortest undirected path from $a$ to $b$, while, for the right-to-left inclusion,
		it is enough to consider the solution $J$ that contains a null value
		for each connected component of $I$, and such that $J$ contains all facts of the
		form $T(a,N)$ for $a\in adom(I)$, where $N$ is the null value associated to the
		connected component of $I$ to which $a$ belongs.

		Now, suppose for the sake of a contradiction that $q_s$ and $\Sigma_s$ as described
		in the statement of the proposition exist. Let $k$ be the number of variables in $q_s$.
		Let I be an instance that consists of a directed path of length $k+1$
		from $a$ to $b$. It follows from the above claim, and from our assumption on $q_s$
		and $\Sigma_s$, that $(a,b)\in \text{subset-}\cqa(q_s,I,\Sigma_s)$, and for every proper
		subinstance $I'$ of $I$, we have
		that $(a,b)\not\in \certain(q,I',M)$.

		\begin{trivlist}
			\item \emph{Claim:} The instance $I$ is consistent with $\Sigma_s$.
		\end{trivlist}

		Suppose for the sake of a contradiction that the above claim does not hold.
		Let $I'$ be any subset-repair of $I$ with respect to
		$\Sigma_s$. Since $I'$ is a proper sub-instance of $I$,
		we have that $(a,b)\not\in \certain(q,I',\Sigma)$.
		In particular, since $I'$ satisfies
		$\Sigma_s$, we have that $(a,b)\not\in q_s(I')$. But since $I'$ is a repair of $I$, this
		means that $(a,b)\not\in \text{subset-}\cqa(q_s,I,\Sigma_s)$, a contradiction.

		Since $(a,b)\in \text{subset-}\cqa(q_s,I,\Sigma_s)$ and $I$ is consistent with $\Sigma_s$ we have
		that $(a,b)\in q_s(I)$. That is, there is a
		homomorphism $h$ from $q_s$ to $I$. Let $I''$ be the sub-instance of $I$ consisting of
		the facts involving only values that are in the image of $h$. Since $I$ contains $k$ facts and $q$ contains
		$k+1$ facts, $I''$ is a proper sub-instance of
                $I$. Moreover, since universal first-order sentences are preserved
		under taking induced sub-instances, every egd true in $I$ is also true in $I''$ and
		therefore, $I''$ is consistent with $\Sigma_s$.
		Finally, by construction, $q_s(I'')=true$. Therefore, $(a,b)\in \text{subset-}\cqa(q_s,I'',\Sigma_s)$.
		This contradicts the fact that $(a,b)\not\in \certain(q,I'',M)$.
	
	\end{proof}


The following result tells us that Theorem~\ref{thm:xr-to-cqa}
also cannot be extended to schema mappings containing GAV target tgds.
	\begin{theorem}\label{thm:reachability-counterexample-gav}
		Consider the \mappingtypethree{\gav}{\gav}{\egd} schema mapping \mapping{}, where
                \begin{itemize}
                  \item		$\mathbf{S}=\{R\}$ and
                    $\mathbf{T}=\{T\}$,
                  \item $\sigst = \{R(x,y) \to T(x,y)\}$, and
                    \item $\sigt = \{T(x,y)\land T(y,z)\to T(x,z)\}$. 
                \end{itemize}
                Consider the query
		$q(x,y) = T(x,y)$ over $\mathbf{T}$.
		There does not exist a UCQ $q_s$ over $\mathbf{S}$ and
                a set
                of universal first-order sentences (in particular, egds)
		$\Sigma_s$ such that, for every instance $I$, we have that
		$\xrcertain(q,I,\sm{M}) = \text{subset-}\cqa(q_s,I,\Sigma_s)$.
	\end{theorem}
	\begin{proof}
		We start by observing that $\certain(q,I,\sm{M})$ expresses directed reachability
		along the relation $R$: for every $\textbf{S}$-instance $I$,
			 $\certain(q,I,\sm{M}) = \{(a,b)\in adom(I) \mid \text{$b$ is reachable from $a$
			 by a directed $R$-path}\}$.
                  The claim is proved by induction on the length of the path.
                 The remainder of the proof is identical to that of
                 Theorem~\ref{thm:reachability-counterexample} (the
                 difference between directed paths and undirected
                 paths is inessential to the argument).
    
    \end{proof}

\section{DLP-Rewritability} \label{sec:dlp}

We saw in the previous section that the applicability of the CQA-rewriting approach is limited to \mappingtypetwo{\gav}{\egd} schema mappings. In 
this section, we consider another approach to computing XR-certain answers, based on a reduction to the problem of computing certain answers
over the stable models of a disjunctive logic program.  Our reduction
is applicable to \mappingtypethree{\glav}{\waglav}{\egd} schema
mappings.
First,  we reduce the case of \mappingtypethree{\glav}{\waglav}{\egd}
schema mappings to the case of \mappingtypethree{\gav}{\gav}{\egd}
schema mappings.

\begin{theorem}\label{thm:glavtogav}
From a \mappingtypethree{\glav}{\waglav}{\egd} schema mapping
$\sm{M}$ we can construct a \mappingtypethree{\gav}{\gav}{\egd}
schema mapping $\hat{\sm{M}}$ such that, from a conjunctive query $q$,
we can construct a union of conjunctive queries $\hat{q}$  with
$\xrcertain(q,I,\sm{M}) = \xrcertain(\hat{q},I,\hat{\sm{M}})$.
\end{theorem}

The proof of Theorem~\ref{thm:glavtogav} is given in Section~\ref{sec:glav-to-gav} (it is entailed by Theorem~\ref{thm:sotgdtogav}).  Theorem~\ref{thm:reachability-counterexample-gav} shows that the CQA-rewriting approach studied in Section~\ref{sec:CQA-rewritability} is, in general, not applicable to \mappingtypethree{\gav}{\gav}{\egd} schema mappings and unions of conjunctive queries. To address this problem, we will now consider a different approach to computing XR-certain answers, using disjunctive logic programs.  Although stable models are popular in the literature, including for database repairs, we find that the selective minimization offered by parallel circumscription is a better fit for \xrcertain{} semantics because our minimality condition applies only to the source-part of the schema.  We then use a result from \cite{DBLP:conf/jelia/JanhunenO04} to translate back into the realm of stable models.

Stable models of disjunctive logic programs have been well-studied as a way to compute database repairs (\cite{DBLP:journals/dke/MarileoB10} provides a thorough treatment).  In~\cite{DBLP:conf/ijcai/CaliLR03}, Cal\`{\i} et al. give an encoding of their loosely-sound semantics for data integration as a disjunctive logic program.  Their encoding is applicable for {\em non-key-conflicting} sets of constraints, a syntactic condition that is orthogonal to {\em weak acyclicity}, and which eliminates the utility of named nulls.  Although their semantics use a notion of minimality that is similar to ours, our setting and our syntactic condition differ sufficiently that our results are complementary.

Fix a domain $Const$. A \emph{disjunctive logic program} (DLP) $\Pi$ over a schema $\mathbf{R}$ is a 
finite collection of rules of the form 
	$$ \alpha_1 \vee \ldots \vee \alpha_n \leftarrow \beta_1,\ldots,\beta_m, \neg\gamma_1,\ldots,\neg\gamma_k. $$
where $n,m,k \ge 0$ and $\alpha_1,\ldots,\alpha_n,\beta_1,\ldots,\beta_m,\gamma_1,\ldots,\gamma_k$ are atoms formed
from the relations in $\mathbf{R}\cup\{=\}$, using the constants in $Const$ and
first-order variables.
A DLP is said to be \emph{positive} if it consists of rules that do not contain
negated atoms except possibly for inequalities.  A DLP is said to be \emph{ground}
if it consists of rules that do not contain any first-order variables.
A \emph{model} of $\Pi$ is an $\mathbf{R}$-instance $I$ over domain $Const$ that
satisfies all rules of $\Pi$ (viewed as universally quantified
first-order sentences).  A rule in which
$n=0$ is called a constraint, and is satisfied only if its body is not satisfied.
A \emph{minimal model} of $\Pi$ is a model $M$ of $\Pi$ such
that there does not exist a model $M'$ of $\Pi$ where
the facts of $M'$ form a strict subset of the facts of $M$.
More generally, for subsets $\mathbf{R}_{\textsc{m}},\mathbf{R}_{\textsc{f}} \subseteq \mathbf{R}$, 
an \emph{$\langle\mathbf{R}_{\textsc{m}},\mathbf{R}_{\textsc{f}}\rangle$-minimal model} of $\Pi$ is a model $M$ of $\Pi$ such
that there does not exist a model $M'$ of $\Pi$ where
the facts of $M'$ involving relations from $\mathbf{R}_{\textsc{m}}$ form a
strict subset of the facts of $M$ involving relations from
$\mathbf{R}_{\textsc{m}}$, and the set of facts of $M'$ involving relations from $\mathbf{R}_{\textsc{f}}$
is equal to the set of facts of $M$ involving relations from $\mathbf{R}_{\textsc{f}}$ \cite{DBLP:conf/jelia/JanhunenO04}.
Although minimal models are a well-behaved semantics for positive
DLPs, it is not well suited for programs with negations. 
 The \emph{stable model} semantics is a widely
used semantics of DLPs that are not necessarily positive. For positive DLPs, it coincides with the
minimal model semantics. For a ground DLP $\Pi$ over a schema $\mathbf{R}$ and an $\mathbf{R}$-instance $M$ over the domain $Const$, the \emph{reduct}
$\Pi^M$ of $\Pi$ with respect to $M$ is the DLP containing, for each
rule $ \alpha_1 \vee \ldots \vee \alpha_n \leftarrow
\beta_1,\ldots,\beta_m, \neg\gamma_1,\ldots,\neg\gamma_k$, with
$M\not\models\gamma_i$ for all $i\leq k$, the rule $ \alpha_1 \vee
\ldots \vee \alpha_n \leftarrow \beta_1,\ldots,\beta_m$.  A
\emph{stable model} of a ground DLP $\Pi$ is an $\mathbf{R}$-instance
$M$ over the domain $Const$ such that $M$ is a minimal model of the reduct $\Pi^M$. See 
\cite{DBLP:conf/iclp/GelfondL88} for more details.

In this section, we will construct positive DLP programs whose 
$\langle\mathbf{R}_{\textsc{m}},\mathbf{R}_{\textsc{f}}\rangle$-minimal models correspond to XR-solutions.
In light of Theorem~\ref{thm:glavtogav}, we may restrict our attention to \mappingtypethree{\gav}{\gav}{\egd}
schema mappings.

In \cite{DBLP:conf/jelia/JanhunenO04} it was shown that a positive ground DLP
$\Pi$ over a schema $\mathbf{R}$, together with subsets $\mathbf{R}_{\textsc{m}},\mathbf{R}_{\textsc{f}} \subseteq\mathbf{R}$, can be translated in polynomial time to a
(not necessarily positive) DLP $\Pi'$ over a possibly larger schema
that includes $\mathbf{R}$,
such that there is a bijection between the
$\langle\mathbf{R}_{\textsc{m}},\mathbf{R}_{\textsc{f}}\rangle$-minimal models of $\Pi$ and the stable models of $\Pi'$,
where every pair of instances that stand in the bijection agree on all facts over the schema
$\mathbf{R}$. This shows that DLP reasoners based on the
stable model semantics, such as DLV \cite{DBLP:journals/tocl/LeonePFEGPS06,DBLP:conf/datalog/AlvianoFLPPT10}, can be used to evaluate positive
ground disjunctive logic programs under the $\langle\mathbf{R}_{\textsc{m}},\mathbf{R}_{\textsc{f}}\rangle$-minimal
model semantics.  Although stated only for ground programs in  
\cite{DBLP:conf/jelia/JanhunenO04}, this technique can be used for arbitrary positive
DLPs through grounding.  Note that, when a program is grounded,
inequalities are reduced to $\top$ or $\bot$.

\begin{theorem}\label{thm:dlp-rewritability-gav}
Given a \mappingtypethree{\gav}{\gav}{\egd} schema mapping  \mapping{},
we can construct in
linear time
a positive DLP $\Pi$ over a schema $\mathbf{R}$ that contains
$\mathbf{S}\cup\mathbf{T}$, and subsets $\mathbf{R}_{\textsc{m}},\mathbf{R}_{\textsc{f}} \subseteq \mathbf{R}$, such that for every union $q$ of conjunctive queries
over $\mathbf{T}$ and for every $\mathbf{S}$-instance $I$, we have that $\xrcertain(q,I,\sm{M}) =
\bigcap\{q(M)\mid M$ is an $\langle\mathbf{R}_{\textsc{m}},\mathbf{R}_{\textsc{f}}\rangle$-minimal model of $\Pi\cup I\}$.
\end{theorem}

\begin{proof}
 We construct a disjunctive logic program $\dlp{\sm{M}}$ for a \mappingtypethree{\gav}{\gav}{\egd} schema mapping \mapping{} as follows:
	\begin{enumerate}
		\item For each source relation $S$ with arity $n$, add the rules
                       \[\begin{array}{l}
 			S_k(x_1,\ldots,x_n) \vee S_d(x_1,\ldots,x_n) \leftarrow S(x_1,\ldots,x_n)\\
			\bot\leftarrow S_k(x_1,\ldots,x_n), S_d(x_1,\ldots,x_n)\\
			S(x_1,\ldots,x_n) \leftarrow S_k(x_1,\ldots,x_n)
                        \end{array}\]
			where $S_k$ and $S_d$ represent the {\em kept} and {\em deleted} atoms of $S$, respectively.
		\item For each s-t tgd $\phi(\mathbf{x}) \to T(\mathbf{x}')$ in \sigst, add the rule
			\[ T(\mathbf{x}') \leftarrow \alpha_1, \ldots, \alpha_m \]
			where $\alpha_1,\ldots,\alpha_m$ are the atoms in $\phi(\mathbf{x})$, in which each relation $S$ has been uniformly replaced by $S_k$.
		\item For each tgd $\phi(\mathbf{x}) \to T(\mathbf{x}')$ in \sigt, add the rule
			\[ T(\mathbf{x}) \leftarrow \alpha_1, \ldots, \alpha_m \]
			where $\alpha_1,\ldots,\alpha_m$  are the atoms in $\phi(\mathbf{x})$.
		\item For each egd $\phi(\mathbf{x}) \to x_1 = x_2$, where $x_1,x_2 \in \mathbf{x}$, add the rule
			\[ \bot\leftarrow \alpha_1, \ldots,  \alpha_m, x_1 \neq x_2, \]
			where $\alpha_1,\ldots,\alpha_m$  are the atoms in $\phi(\mathbf{x})$.
	\end{enumerate}

We minimize the model w.r.t. $\mathbf{R}_{\textsc{m}} = \{ S_d \mid S \in \mathbf{S}\}$, and fix $\mathbf{R}_{\textsc{f}} = \{ S \mid S \in \mathbf{S} \}$.  The disjunctive logic program for \sm{M}, denoted $\dlp{\sm{M}}$, is a straightforward encoding of the constraints in \sigst{} and \sigt{} as disjunctive logic rules over an indefinite view of the source instance.  Since the source instance is fixed, the rules of the form $ S(x_1,\ldots,x_n) \leftarrow S_k(x_1,\ldots,x_n) $ in $\dlp{\sm{M}}$ force the kept atoms to be a sub-instance of the source instance.  Notice that egds are encoded as denial constraints, and that disjunction is used only to non-deterministically choose a subset of the source instance.

To prove the theorem, we first show that the restriction of every $\langle\mathbf{R}_{\textsc{m}},\mathbf{R}_{\textsc{f}}\rangle$-minimal model of $\dlp{\sm{M}} \cup I$ to the schema $\{S_k\mid S\in\mathbf{S}\}\cup\mathbf{T}$ constitutes an exchange-repair solution.  We then show that for every exchange-repair solution, we can build a corresponding $\langle\mathbf{R}_{\textsc{m}},\mathbf{R}_{\textsc{f}}\rangle$-minimal model of $\dlp{\sm{M}} \cup I$.

		Let \mapping{} be a \mappingtypethree{\gav}{\gav}{\egd} schema mapping.
		Let $\Pi = \dlp{\sm{M}}$.  Let $\mathbf{R}$ be the schema of $\Pi$, and let $\mathbf{R}_{\textsc{m}} = \{ S_d \mid S \in \mathbf{S}\}$, and $\mathbf{R}_{\textsc{f}} = \{ S \mid S \in \mathbf{S} \}$.
		Let $q$ be a union of conjunctive queries over $\mathbf{T}$, and let $I$ be an $\mathbf{S}$-instance.

		We first prove that a certain restriction of every $\langle\mathbf{R}_{\textsc{m}},\mathbf{R}_{\textsc{f}}\rangle$-minimal model of $\Pi \cup I$ is an exchange-repair solution.  
		Let $M$ be an $\langle\mathbf{R}_{\textsc{m}},\mathbf{R}_{\textsc{f}}\rangle$-minimal model of $\Pi \cup I$.  Then for each source atom $S(c_1,\ldots,c_n)$, $M$ satisfies $S_k(c_1,\ldots,c_n) \vee S_d(c_1,\ldots,c_n) \leftarrow S(c_1,\ldots,c_n)$ and $\leftarrow S_k(c_1,\ldots,c_n), S_d(c_1,\ldots,c_n)$ and therefore contains exactly one of $S_k(c_1,\ldots,c_n)$ or $S_d(c_1,\ldots,c_n)$.  Furthermore, for every atom $S_k(d_1,\ldots,d_n) \in M$, $M$ satisfies $S(d_1,\ldots,d_n) \leftarrow S_k(d_1,\ldots,d_n)$.  Let $I'$ be a renaming of the restriction of $M$ to the {\em kept} predicates (by removal of the $_k$ subscript), and observe that since $I$ is fixed, $I'$ is a sub-instance of $I$.  Furthermore, since $M \models \Pi \cup I$ (which contains copies of the constraints of \sm{M} over its {\em kept} predicates), $I'$ has a solution w.r.t. \sm{M}.  Finally, an appropriate renaming (by removal of the $_d$ subscript) of the restriction of $M$ to $\mathbf{R}_{\textsc{m}}$ (the {\em deleted} predicates) is equal to $I \setminus I'$, and since $M$ is a $\langle\mathbf{R}_{\textsc{m}},\mathbf{R}_{\textsc{f}}\rangle$-minimal model of $\Pi \cup I$, we have that there is no $I''$ such that $I' \subset I'' \subseteq I$ and a solution exists for $I''$ w.r.t. \sm{M}.  Therefore, $I'$ is a source repair of $I$ w.r.t. \sm{M}, and $I'$ along with the restriction of $M$ to $\mathbf{T}$ is an exchange-repair solution for $I$ w.r.t. \sm{M}.

		We now prove that for every exchange-repair solution, there exists an $\langle\mathbf{R}_{\textsc{m}},\mathbf{R}_{\textsc{f}}\rangle$-minimal model of $\Pi \cup I$.  
		Let $(I',J')$ be an exchange-repair solution for $I$ w.r.t \sm{M}.  Let $M = I \cup I'_k \cup (I \setminus I')_d \cup J'$, where $I'_k$ is $I'$ renamed over the {\em kept} predicates, and $(I \setminus I')_d$ is $I \setminus I'$ renamed over the {\em deleted} predicates.  Since $I'$ is a subset of $I$, and $I \setminus I'$ is disjoint from $I'$, we have that the rules of the forms $ S_k(x_1,\ldots,x_n) \vee S_d(x_1,\ldots,x_n) \leftarrow S(x_1,\ldots,x_n), $ and $ \leftarrow S_k(x_1,\ldots,x_n), S_d(x_1,\ldots,x_n), $ and $ S(x_1,\ldots,x_n) \leftarrow S_k(x_1,\ldots,x_n) $ are satisfied.  It also holds that $(I',J')$ satisfy \sm{M}, and therefore $M$ is a model of $\Pi \cup I$.  Finally, since there is no $I''$ such that $I' \subset I'' \subseteq I$ and a solution exists for $I''$ w.r.t. \sm{M}, $M$ is also $\langle\mathbf{R}_{\textsc{m}},\mathbf{R}_{\textsc{f}}\rangle$-minimal.
	
		Therefore, $\xrcertain(q,I,\sm{M}) = \bigcap\{q(M)\mid M$ is an $\langle\mathbf{R}_{\textsc{m}},\mathbf{R}_{\textsc{f}}\rangle$-minimal model of $\Pi\cup I\}$.
	
	\end{proof}
	
	Figure~\ref{fig:DLP} illustrates the disjunctive logic program obtained for the schema mapping from Figure~\ref{example:mapping}.
	
\begin{figure*}
\eat{
    Schema Mapping and Query:\\
	\begin{tabular}{|lll|}
		\hline
		\sigst & \hspace{-2mm}$=$ & \hspace{-2mm}$\left\{ \begin{array}{lllll} {\tt Task\_Assignments}(p,t,d) & \to & {\tt Departments}(p,d) \wedge {\tt Tasks}(p,t) \\
															{\tt Stakeholders\_old}(t,s) & \to & {\tt Stakeholders\_new}(t,s) \end{array} \right\}$ \\
		\sigt & \hspace{-2mm}$=$ & \hspace{-2mm}$\left\{ \begin{array}{l} {\tt Departments}(p,d) \wedge {\tt Departments}(p,d') \to d=d' \end{array} \right\}$\\
		\multicolumn{3}{|l|}{${\tt boss}(person,stakeholder) = \exists task ({\tt Tasks}(person, task) \wedge {\tt Stakeholders\_new}(task,stakeholder))$}\\
		\hline
	\end{tabular}
	\vspace{2mm}
	
	\vspace{2mm} Disjunctive Logic Program:\\
}
\centering\normalsize
	\fbox{\parbox{15.5cm}{
		${\tt Task\_Assignments}_k(p,t,d) \vee {\tt Task\_Assignments}_d(p,t,d) \leftarrow {\tt Task\_Assignments}(p,t,d).$ \\
		~~$\bot \leftarrow {\tt Task\_Assignments}_k(p,t,d), {\tt Task\_Assignments}_d(p,t,d).$ \\
		~~${\tt Task\_Assignments}(p,t,d) \leftarrow {\tt Task\_Assignments}_k(p,t,d).$ \\
		~\\
		~~${\tt Stakeholders\_old}_k(t,s) \vee {\tt Stakeholders\_old}_d(t,s) \leftarrow {\tt Stakeholders\_old}(t,s).$ \\
		~~$\bot \leftarrow {\tt Stakeholders\_old}_k(t,s) \wedge {\tt Stakeholders\_old}_d(t,s).$ \\
		~~${\tt Stakeholders\_old}(t,s) \leftarrow {\tt Stakeholders\_old}_k(t,s).$ \\
		~\\
		~~${\tt Departments}(p,d) \leftarrow {\tt Task\_Assignments}_k(p,t,d).$ \\
		~~${\tt Tasks}(p,t) \leftarrow {\tt Task\_Assignments}_k(p,t,d).$ \\
		~~${\tt Stakeholders\_new}(t,s) \leftarrow {\tt Stakeholders\_old}_k(t,s).$ \\
		~\\
		~~$\bot \leftarrow {\tt Departments}(p,d), {\tt Departments}(p,d'), d \neq d'.$\\
		~\\
		~~${\tt boss}(person,stakeholder) \leftarrow {\tt Tasks}(person,task), {\tt Stakeholders\_new}(task,stakeholder)$
	}}
	\caption{The disjunctive logic program over $\langle\mathbf{R}_{\textsc{m}},\mathbf{R}_{\textsc{f}}\rangle$-minimal models for the schema mapping and query given in Figure~\ref{example:mapping}.  Here, $\mathbf{R}_{\textsc{m}}~=~\{{\tt Task\_Assignments}_d,{\tt Stakeholders\_old}_d\}$ and $\mathbf{R}_{\textsc{f}}~=~\{{\tt Task\_Assignments},{\tt Stakeholders\_old}\}$ }
	\label{fig:DLP}
\end{figure*}

\section{From \mappingtypethree{\glav}{\waglav}{\egd} to \mappingtypethree{\gav}{\gav}{\egd} }\label{sec:glav-to-gav}

\newcommand{\ucqreduces}{\rightsquigarrow_\textup{UCQ}}

In this section, we prove Theorem~\ref{thm:glavtogav}, and discuss some additional literature related to this particular result.
Let $\sm{M}_1$ and $\sm{M}_2$ be schema mappings with the same source schema. We will write 
$\sm{M}_1\ucqreduces\sm{M}_2$ if for every UCQ $q$ over the target schema of $\sm{M}_1$, there is a UCQ $q'$ over the 
target schema of $\sm{M}_2$ such that for all source instance $I$, $\xrcertain(q,I,\sm{M}_1) = \xrcertain(q',I,\sm{M}_2)$.
Using this notation, 
Theorem~\ref{thm:glavtogav} states that \emph{for every \mappingtypethree{\glav}{\waglav}{\egd} schema mapping $\sm{M}$ there is a \mappingtypethree{\gav}{\gav}{\egd}  schema mapping $\sm{M}'$
with $\sm{M}\ucqreduces\sm{M}'$.}
We will in fact prove a stronger statement that applies to schema mappings defined by \emph{second-order tgds}.  Second-order tgds serve not only to strengthen the result, but also to make its proof more natural.

\subsection{Second-order TGDs}

Second-order tgds are a natural extension of tgds that was introduced in \cite{DBLP:journals/tods/FaginKPT05} 
in the context of schema mapping composition. We recall the definition.

Let ${\bf f}$ be a collection of function symbols, each having a designated arity.  A {\em simple term} is a constant or variable.  A {\em compound term} is a function applied to a list of terms, such that the arity of the function symbol is respected.  By an {\em $\bf f$-term}, we mean either a simple term, or a compound term built up from variables and/or constants using the function symbols in $\bf f$.  We will omit $\bf f$ from the notation when it is understood from context.
The \emph{depth} of a term is the maximal nesting of function symbols, with {\em depth}$(e)=0$ when $e$ is a simple term. A {\em ground term} is a term in which no variables appear.

A \emph{second-order tgd} (SO tgd) over a schema $\bf R$ is an expression of the form 
\[\sigma = \exists {\bf f} \left(\forall {\bf x_1}(\phi_1 \to \psi_1) \wedge \cdots \wedge \forall {\bf x_n}(\phi_n \to \psi_n) \right)\]
 where ${\bf f}$ is a collection of function symbols, and
 \begin{enumerate}
 \item each $\phi_i$ is a conjunction of (a) atoms $S(y_1, \ldots, y_k)$ where $S\in\bf R$ and $y_1, \ldots, y_k$ are variables from $\bf x_i$; and (b) equalities of the form $t_1=t_2$ where $t_1, t_2$ are terms over $\bf x_i$ and $\bf f$.
 \item each $\psi_i$ is a conjunction of atoms $S(t_1, \ldots, t_k)$ where $S\in \bf R$ and $t_1, \ldots, t_k$ are $\bf f$-terms built from $\bf x_i$.  
 \item each variable in $\bf x_i$ occurs in a relational atom in $\phi_i$. 
 \end{enumerate}
We say that an $\bf R$-instance $I$ satisfies $\sigma$ if there exists a collection of functions $\bf f^0$ (whose
domain and range are $\Const \cup \Nulls$) such that each ``clause'' $\forall {\bf x_i}(\phi_i \to \psi_i)$
of $\sigma$ is satisfied in $I$ where each function symbol in $\bf f$ is interpreted by the corresponding function in $\bf f^0$. 
We will write $I\models\sigma$ when this is the case, or, if we wish to make $\bf f^0$ explicit in the notation, 
 $I\models\sigma~[\bf f \mapsto \bf f^0]$.

A \emph{source-to-target} SO tgd for source schema $\bf S$ and target schema $\bf T$ is an SO tgd over $\bf S \cup \bf T$, of the above form, where each $\phi_i$ contains only relation symbols from $\bf S$ and each $\psi_i$ contains only relation symbols from $\bf T$. We note that, in \cite{DBLP:journals/tods/FaginKPT05}, only source-to-target SO tgds were considered.

An \emph{equality-free} SO tgd (\efsotgd) is an SO tgd that does not contain term equalities. 
We denote by \mappingtypethree{\sotgd}{\sotgd}{\egd} the class of schema mappings \mapping where $\Sigma_{st}$ is a set of
source-to-target SO tgds over $\bf S$ and $\bf T$, and $\Sigma_t$ is a set of SO tgds and/or egds over $\bf T$.
Other classes of schema mappings, such as \mappingtypetwo{\sotgd}{\sotgd} and 
\mappingtypetwo{\efsotgd}{\efsotgd}, are defined analogously.  Note that an important subclass of equality-free SO tgds are the {\em plain} SO tgds, introduced in~\cite{DBLP:journals/jcss/Arenas0RR13}, in which no terms contain nested functions.

It is  known that every tgd is logically equivalent to a SO tgd, which can be obtained from it by \emph{skolemization}
\cite{DBLP:journals/tods/FaginKPT05}. Although stated in the literature only for the case of source-to-target tgds \cite{DBLP:journals/tods/FaginKPT05}, the same applies to target tgds. 
Figure~\ref{fig:skolemized} shows the skolemization of the example schema mapping in Figure~\ref{fig:example}.

Moreover, if we adapt the concept of weak acyclicity to SO tgds in the
appropriate way, then every weakly acyclic set of tgds is logically equivalent to a weakly acyclic SO tgd. 

More precisely, we say that a set $\Sigma$ of SO tgds is weakly acyclic if there is no cycle in its dependency graph containing a special edge,
where the dependency graph associated to a set of SO tgds is defined as follows:
\begin{enumerate}
\item  the directed graph whose nodes are positions $(R,i)$ where $R$ is a relation symbol and $i$ is an attribute position of $R$ (as before) 
\item  there is a normal edge from $(R,i)$ to $(S,j)$ if $\Sigma$ contains a SO tgd of the form 
    \[\sigma = \exists {\bf f} \left(\forall {\bf x_1}(\phi_1 \to \psi_1) \wedge \cdots \wedge \forall {\bf x_n}(\phi_n \to \psi_n) \right)\]
    and for some $i\leq n$, $\phi_i$ contains a variable in position $(R,i)$ and $\psi_i$ contains the same variable in position $(S,j)$.
\item  there is a special edge from $(R,i)$ to $(S,j)$ if $\Sigma$ contains a SO tgd of the form 
    \[\sigma = \exists {\bf f} \left(\forall {\bf x_1}(\phi_1 \to \psi_1) \wedge \cdots \wedge \forall {\bf x_n}(\phi_n \to \psi_n) \right)\]
    and for some $i\leq n$, $\phi_i$ contains a variable in position $(R,i)$ and $\psi_i$ contains a compound term in position $(S,j)$ containing the same variable. 
\end{enumerate}
We then have: 

\begin{figure}[t]
	\centering
	\begin{tabular}{|lll|}
		\hline
		$\bf S$ & $=$ & $\left\{ \rel R \right\}$ \\
		$\bf T$ & $=$ & $\left\{ \rel T \right\}$ \\
		\sigst & $=$ & $\left\{ \begin{array}{lllll} \rel R(x,y) & \to & \exists u (\rel T(x,u) \wedge \rel T(y,u))\end{array} \right\}$ \\
		\sigt & $=$ & $\left\{ \begin{array}{l} \rel T(x,y) \wedge \rel T(x,y') \to y=y'\end{array} \right\}$\\
		\multicolumn{3}{|l|}{$\rel{q}(x,y) \defeq \exists u (\rel T(x, u) \wedge \rel T(y,u))$}\\
		\hline
	\end{tabular}
    \caption{An example schema mapping and query.}
    \label{fig:example}
\end{figure}

\begin{figure}[t]
	\centering
	\begin{tabular}{|lll|}
		\hline
		$\bf S$ & $=$ & $\left\{ \rel R \right\}$ \\
		$\bf T$ & $=$ & $\left\{ \rel T \right\}$ \\
		\sigst & $=$ & $\left\{ \exists f \left( \begin{array}{lllll}
			\rel R(x,y) & \to & \rel T(x,f(x,y)) \wedge \\
			\rel R(x,y) & \to & \rel T(y,f(x,y))
		\end{array} \right) \right\}$ \\
		\sigt & $=$ & $\left\{ \begin{array}{l} \rel T(x,y) \wedge \rel T(x,y') \to y=y'\end{array} \right\}$\\
		\multicolumn{3}{|l|}{${\tt q}(x,y) \defeq \exists u (\rel T(x, u) \wedge \rel T(y,u))$}\\
		\hline
	\end{tabular}
    \caption{Result of skolemizing the schema mapping in Figure~\ref{fig:example}.}
    \label{fig:skolemized}
\end{figure}

\begin{proposition}
Every  \mappingtypethree{\glav}{\waglav}{\egd} schema mapping is logically equivalent to a weakly acyclic \mappingtypethree{\efsotgd}{\efsotgd}{\egd} schema mapping.
\end{proposition}
Indeed, if $\sm{M}$ is a \mappingtypethree{\glav}{\waglav}{\egd} schema mapping, and $\sm{M}'$ is the \mappingtypethree{\efsotgd}{\efsotgd}{\egd}  schema mapping obtained from $\sm{M}$ by skolemization,
then $\sm{M}$ and $\sm{M'}$ are logically equivalent. Moreover, it is easy to see that $\sm{M}$ and $\sm{M}'$ have the same dependency graph, and, therefore, 
$\sm{M}'$ is weakly acyclic. 



In the remainder of this section, we will establish:

\begin{theorem}\label{thm:sotgdtogav}
For every weakly acyclic \mappingtypethree{\sotgd}{\sotgd}{\egd} schema mapping $\sm M$ there is a 
\mappingtypethree{\gav}{\gav}{\egd} schema mapping $\sm M'$ such that $\sm M\ucqreduces \sm M'$.
\end{theorem}

The proof borrows ideas from previous literature, and we discuss  relevant related work at the end of the section.


\subsection{Eliminating Equalities to Establish Freeness}
In this section, we will rewrite our schema mapping to eliminate egds as well as equality conditions in SO tgds. This allows us to work with solutions in which there is a one-to-one correspondence between ground terms (of any depth) and their values.  This property, called {\em freeness}, which we define below, is used in Section~\ref{sec:skeleton_step}.

For simplicity we first restrict attention to \mappingtypethree{\efsotgd}{\sotgd}{\egd} schema mappings.

\begin{definition}[Equality Singularization]
	Fix a fresh binary relation symbol Eq.

	\begin{itemize}
	\item 
	The \emph{equality singularization} of a conjunctive query $q({\bf x}) = \exists {\bf y} \phi({\bf x},{\bf y})$, denoted by 
    $q^{\rm Eq}({\bf x})$, is the conjunctive query $\exists {\bf y}{\bf z} \phi'({\bf x},{\bf y},{\bf z})$ obtained from $q$ as follows:
	whenever a variable $u$  (free or quantified) occurs more than once in $\phi$, we replace each occurence other than the first occurrence
	by a fresh distinct variable $z$ and we add the atom $Eq(u,z)$.

	\item 
	The \emph{equality singularization} of an egd 
	    \[\sigma = \forall {\bf x}(\phi\to x_i=x_j)\]
	is the GAV tgd
	    \[\sigma^{Eq} = \forall {\bf x}{\bf z}(\phi'\to Eq(x_i,x_j))\]
	 where $\phi^{Eq} = \exists {\bf z}\phi'$

	\item 
	The \emph{equality singularization} of an SO tgd 
	    \[\sigma = \exists {\bf f} \bigwedge_{i=1\ldots n}\left(\forall {\bf x}_i(\phi_i\land\alpha_i \to \psi_i)\right)\]
	 (where each $\phi_i$ is a conjunction of relational atoms and each $\alpha_i$ is a conjunction of equalities) is the equality-free SO tgd 

	    \[\sigma' = \exists {\bf f} \bigwedge_{i=1\ldots n}\left(\forall {\bf x}_i {\bf z}_i (\phi_1'\land\alpha_i' \to \psi_i) \right)\]

	 where $\phi_i^{Eq} = \exists {\bf z}_i\phi'$ and $\alpha'_i$ is obtained from $\alpha_i$ by replacing each equality $s=t$ by 
	 $Eq(s,t)$.

	 \item 
	The \emph{equality singularization} of 	a \mappingtypethree{\efsotgd}{\sotgd}{\egd} schema mapping  \mapping is
	the \mappingtypetwo{\efsotgd}{\efsotgd} schema mapping 
\[{\cal M}^{\rm Eq} = ({\bf S},{\bf T} \cup \{ \rel{Eq} \}, \sigst, \{ \sigma^{\rm Eq} \mid \sigma \in \sigt \} \cup \text{eqAx}({\bf T}))\]
	where eqAx$({\bf T})$ is the set of (full) tgds of the form $\rel T(x_1,\ldots,x_n) \to \rel{Eq}(x_1,x_1) \wedge \cdots \wedge \rel{Eq}(x_n,x_n)$ where $\rel T$ is a relation in $\bf T$, along with the tgds $\rel{Eq}(x_1,x_2) \to \rel{Eq}(x_2,x_1)$ and $\rel{Eq}(x_1,x_2) \wedge \rel{Eq}(x_2,x_3) \to \rel{Eq}(x_1,x_3)$.
\end{itemize}
\end{definition}

\begin{figure}[t]
	\centering
	\begin{tabular}{|lll|}
		\hline
		$\bf S$ & $=$ & $\left\{ \rel R \right\}$ \\
		$\bf T$ & $=$ & $\left\{ \rel T, \rel{Eq} \right\}$ \\
		\sigst & $=$ & $\left\{ \exists f \left( \begin{array}{lllll}
			\rel R(x,y) & \to & \rel T(x,f(x,y)) \wedge \\
			\rel R(x,y) & \to & \rel T(y,f(x,y))
		\end{array} \right) \right\}$ \\
		$\sigt^{\rm Eq}$ & $=$ & $\left\{ \begin{array}{l}
			\rel T(x,y) \wedge \rel{Eq}(x,x') \wedge \rel T(x',y') \to \rel{Eq}(y,y')\\
			\rel T(x,y) \to \rel{Eq}(x,x)\\
			\rel T(x,y) \to \rel{Eq}(y,y)\\
			\rel{Eq}(x_1,x_2) \to \rel{Eq}(x_2,x_1)\\
			\rel{Eq}(x_1,x_2) \wedge \rel{Eq}(x_2,x_3) \to \rel{Eq}(x_1,x_3)\\
		\end{array} \right\}$\\
		\multicolumn{3}{|l|}{$\rel{q}^{\rm Eq}(x,y) \defeq \exists u,u' (\rel T(x, u) \wedge \rel{Eq}(u,u') \wedge \rel T(y,u'))$}\\
		\hline
	\end{tabular}
    \caption{Equality singularization of the schema mapping and query from Figure~\ref{fig:skolemized}.}
    \label{fig:singularized}
\end{figure}

Figure~\ref{fig:singularized} shows equality singularization in action.  

\begin{proposition}\label{prop:singularization}
	Let \mapping be an \mappingtypethree{\sotgd}{\sotgd}{\egd} schema mapping and let
	 ${\cal M}^{\rm Eq}$ be its equality singularization.  
	For every $\bf S$-instance $I$,
	\begin{enumerate}
	\item
	$I$ has a solution w.r.t. $\sm{M}$  if and only if $I$ has a solution $J'$ w.r.t. ${\cal M}^{\rm Eq}$ such that there is no pair of distinct constants $a,b$ where $J' \models \rel{Eq}(a,b)$. 
	\item If $I$ has a solution w.r.t.~$\sm{M}$, then for every UCQ $q$ over $\bf T$, 
	$\certain(q,I,\sm{M})= \certain(q^{\rm Eq},I,\sm{M}^{\rm Eq})$.
	\end{enumerate}
\end{proposition}

\begin{proof}{}

[$\Rightarrow$]
	Let $J$ be any $\bf T$-instance that is a solution for $I$ with respect to $\sm{M}$. 
	Take $J'$ to be the ${\bf T}\cup\{Eq\}$-instance that extends $J$ with all facts of
	the form $Eq(a,a)$ with $a\in adom(J)$. It is easy to see that $J'$ is a solution
	for $I$ with respect to $\sm{M}^{\rm Eq}$, and that, for all UCQs $q$, we have that
	$q{\downarrow(J)} = q^{\rm Eq}{\downarrow}(J')$. Moreover, it is 
	immediate from the construction of $J'$ that there is no pair of distinct constants
	$a, b$ where $J'\models Eq(a,b)$.

[$\Leftarrow$]
	Let $\bf f$ be the collection of function symbols appearing in ${\cal M}^{\rm Eq}$.
	Let $J$ be a ${\bf T} \cup \{\rel{Eq}\}$-instance that is a solution for $I$ with respect to $\sm{M}^{\rm Eq}$, such that there is no pair of distinct constants $a,b$ where $J \models \rel{Eq}(a,b)$. 
	Note that $\rel{Eq}$ is an equivalence relation and that each equivalence class contains at most one constant (but possibly many null values).
	Let $\bf f^0$ be a witnessing collection of functions, such that $(I,J) \models {\cal M}^{\rm Eq}~[{\bf f}\mapsto {\bf f^0}]$. 
	 We will construct a ${\bf T}$-instance $J'$ and a collection $\bf f^1$ of functions, as follows:
	\begin{itemize}
		\item For every $\rel{Eq}$-equivalence class, choose a single representative member. 
				If an equivalence class contains a constant, we use that constant as the representative member.
				For every value $u\in adom(J)$, denote by $\pi(u)$ the representative member of the $\rel{Eq}$-equivalence class to which $u$ belongs.
		\item $J'$ contains, for every fact $\rel T(v_1,\ldots, v_n)$ of $J$ (where $\rel T \in \bf T$), 
		the corresponding fact $\rel T(\pi(v_1),\ldots,\pi(v_n))$ 
		\item $\bf f^1$ contains, for each function $f$ in $\bf f^0$, the corresponding function $f'$ given by $f'({\bf u})=\pi(f({\bf u}))$.  
	\end{itemize}
	By construction, we have that, for any image $q^{\rm Eq}({\bf a})$ in $J$ of the equality singularization of a conjunctive query $q({\bf x})$, we have an image $q({\bf a})$ in $J'$, and vice versa. This tells us both that $J'$ is a solution for $I$ w.r.t.~\sm{M} and that for any UCQ $q$ over $\bf T$, we have $q{\downarrow}(J') = q^{\rm Eq}{\downarrow}(J)$.    Additionally, since each $\rel{Eq}$-class is represented by a single member in $J'$, we have that,
		for each egd in $\sigma\in\sigt$, the fact that $J$ satisfies $\sigma^\rel{Eq}$
		 implies that $J'$ satisfies $\sigma$. 
\end{proof}

The importance of Proposition~\ref{prop:singularization} comes from the following observation.
Consider any \mappingtypetwo{\sotgd}{\sotgd} schema mapping \mapping. 
Let ${\bf f}$ be the collection of all function symbols occurring in SO tgds in $\Sigma_{st}\cup\Sigma_t$.
A solution $J$  for a source instance $I$ with respect to $\sm{M}$ is 
said to be a \emph{free} solution if there is a collection of functions ${\bf f^0}$ such that
$(I,J)\models\Sigma_{st}\cup\Sigma_t~[{\bf f}\mapsto {\bf f^0}]$, and such that
each function in ${\bf f^0}$
   is injective and the functions all have mutually disjoint ranges.
Equivalently, in a free solution, each value in $adom(J)$ is the denotation of 
exactly one ground term.  If, furthermore, we have that each value in $adom(J)$ is the denotation of a (unique) term of depth $k$, then we say that $J$ is a \emph{free solution of rank $k$.}



	
	

\begin{proposition}\label{prop:free-universal}
Let $\sm{M}$ be the equality singularization of a weakly acyclic \mappingtypetwo{\efsotgd}{\sotgd} schema mapping.
There is a natural number $k\geq 0$ such that
every source instance $I$ has a free universal solution $J$ of rank $k$.
\end{proposition}

\begin{proof} (sketch)
 	Let ${\bf f^0}$ be an arbitrary collection of injective and mutually range-disjoint functions. 
 	Let $J$ be the result of chasing $I$ with the SO tgds of $\sm{M}$ using these functions. 
 	A priori, $J$ is potentially infinite.
 	However, we can show that $J$ is always finite, moreover, of finite rank. 
 	 This is proved by induction: we associate to each position $(R,i)$ (where $R$ is a relation symbol and $i$ an attribute of $R$) a rank, namely
 	 the maximal number of special edges on an incoming path to $(R,i)$ in the dependency graph \emph{times} the maximal depth of a term occurring in the right-hand side of an SO tgd.  Then, we can prove
 	 by  a straightforward induction on $k$ that for all positions $(R,i)$ of rank $k$, each value in position $(R,i)$ is the denotation
 	 of a term of depth at most $k$.  
\end{proof}

It is not hard to see that the same does not hold in the presence of egds.

The above definition of $\sm{M}^{Eq}$ applies only to \mappingtypethree{\efsotgd}{\sotgd}{\egd} schema mappings.
However, it can be extended to  arbitrary \mappingtypethree{\sotgd}{\sotgd}{\egd} schema mappings
\mapping as follows: from $\sm{M}$, we first construct a schema mapping
$\sm{M}' = ({\bf S}, {\bf T}', \Sigma_{copy}, \Sigma_t')$
where ${\bf T}' = {\bf T}\cup\{R'\mid R\in {\bf S}\}$; $\Sigma_{copy} = \{ \forall{\bf x}(R({\bf x})\to R'({\bf x}))\mid R\in {\bf S}\}$;
and $\Sigma_t' = \Sigma_t \cup \{\sigma'\mid \sigma\in \Sigma_{st}\}$, where $\sigma'$ is a copy of $\sigma$ in which every
occurrence of a relation $R\in{\bf S}$ is replaced by $R'$. We then define the equality singularization $\sm{M}^{Eq}$ to be the 
equality singulatization of $\sm{M}'$. It is easy to see that Proposition~\ref{prop:singularization} and Proposition~\ref{prop:free-universal} then hold true for 
arbitrary \mappingtypethree{\sotgd}{\sotgd}{\egd} schema mappings.

\subsection{The Skeleton Rewriting Step} \label{sec:skeleton_step}
Suppose $\sm{M}$ is a weakly acyclic \mappingtypetwo{\efsotgd}{\efsotgd} schema mapping, and ${\cal M}^{\rm Eq}$ is the equality singularization of ${\cal M}$.  Since ${\cal M}^{\rm Eq}$ admits free universal solutions, we can represent the value of every compound term simply by its syntax.  This makes it possible to rewrite ${\cal M}^{\rm Eq}$ in such a way that the syntax of compound terms is captured using specialized relations, and constraints with only simple terms.

The \emph{skeleton} of a term is the expression obtained by replacing all constants and variables
by $\bullet$, where $\bullet$ is a fixed symbol that is not a function symbol~\cite{DBLP:journals/corr/abs-1106-3745}. Thus, for example, the skeleton of $f(g(x,y),z)$ is $f(g(\bullet,\bullet),\bullet)$. 
The \emph{arity} of a skeleton $s$, denoted by $arity(s)$, is the number of occurrences of $\bullet$,
and the depth of a skeleton is defined in the same way as for terms.
If $s,s'_1, \ldots, s'_k$ are skeletons with $arity(s)=k$, then we denote by 
$s(s'_1, \ldots, s'_k)$ the skeleton of arity $arity(s'_1)+\cdots+arity(s'_k)$ obtained
by replacing, for each $i\leq k$,  the $i$-th occurrence of $\bullet$ in $s$ by $s'_i$. 

\begin{definition}
	Let \mapping be a  weakly acyclic \mappingtypetwo{\efsotgd}{\efsotgd} schema mapping with rank $r$ and whose most deeply nested term has depth $d$.  Let $\Theta$ be the set of functions appearing in \sigt.
	Define the {\em skeleton rewriting} of \sm{M} as the schema mapping ${\cal M}^{\rm skel} = ({\bf S},{\bf T}^{\rm skel},\sigst^{\rm skel},\sigt^{\rm skel})$, where:
	\begin{itemize}
		
		\item For every $n$-ary relation $\rel T \in {\bf T}$, let ${\bf T}^{\rm skel}$ contain all relations of the form $\rel T_{s_1,\ldots,s_n}$, where $s_1,\ldots s_n$ are skeletons of depth less than or equal to $r$.
		
		\item For every clause $\phi({\bf x}) \to \rel T(\tau_1,\ldots,\tau_n)$ of a s-t \efsotgd{} in \sigst, let $\sigst^{\rm skel}$ contain the s-t tgd $\phi({\bf x}) \to \rel T_{s_1,...,s_n}(\bar{\bf x})$, where $s_1,\ldots,s_n$ are the skeletons for $\tau_1,\ldots,\tau_n$ respectively, and $\bar{\bf x}$ is the sequence of variables in $\tau_1,\ldots,\tau_n$.
		
		\item For every clause $\phi({\bf x}) \to \rel T(\tau_1,\ldots,\tau_n)$ of a \efsotgd{} in \sigt (where ${\bf x}=x_1, \ldots, x_m$), let $\sigt^{\rm skel}$ contain the tgd $\phi_{s_1, \ldots, s_m}({\bf y}_1, \ldots, {\bf y}_m) \to \rel T_{s'_1, \ldots, s'_n}({\bf \bar{y}_1, \ldots, \bar{y}_n})$, where
			\begin{itemize}
				
				\item $s_1, \ldots, s_m$ are $\Theta$-skeletons of depth at most $r*d$;
				
				\item each ${\bf y}_i$ is a sequence of $\mbox{arity}(s_i)$ fresh variables;
				
				\item $\phi_{s_1, \ldots, s_m}({\bf y}_1, \ldots, {\bf y}_m)$ is obtained from $\phi$ by replacing each atom $\rel R(x_1, \ldots, x_l)$ by $\rel R_{s_1, \ldots, s_l}({\bf y}_1, \ldots, {\bf y}_l)$;
				
				\item $s'_i$ is a $\Theta$-skeleton of depth at most $r*d$ such that $s'_i=s_k$ (if $\tau_i$ is the term $x_k$) and $s'_i=t_i(s_1, \ldots, s_m)$ (if $\tau_i$ is the term $t_i(x_1,\ldots,x_m)$)
				
				\item ${\bf \bar{y}_i} = ({\bf y}_{k_1}, \ldots, {\bf y}_{k_{arity(s_k)}})$ (if $\tau_i$ is the term $x_k$) and ${\bf \bar{y}_i} = (y_{1_1}, ..., y_{1_{arity(s_1)}}, ..., y_{m_1}, ..., y_{m_{arity(s_m)}})$ (if $\tau_i$ is the term $t_i(x_1,\ldots,x_m)$)
				
			\end{itemize}
	\end{itemize}
	In addition, for each conjunctive query $q({\bf x})=\exists{\bf y}\psi({\bf x},{\bf y})$ over ${\bf T}$ with ${\bf x}=x_1, \ldots, x_n$ and ${\bf y}=y_1, \ldots, y_m$, we denote by $q^{\rm skel}({\bf x})$ the union of conjunctive queries over ${\bf T}^{\rm skel}$ of the form $\exists {\bf z}_1\ldots{\bf z}_m\psi_{s_1, \ldots, s_n, s'_1, \ldots, s'_m}(x_1, \ldots, x_n, {\bf z}_1, \ldots, {\bf z}_m)$, where $s_1=\ldots=s_n=\bullet$; $s'_1, \ldots, s'_m$ are $\Theta$-skeletons of depth at most $r$; and each ${\bf z}_i$ is a sequence of fresh variables of length $arity(s'_i)$. 
\end{definition}

For example, consider the \mappingtypetwo{\efsotgd}{\efsotgd} schema mapping $\sm{M}$ whose constraints are $\exists f \forall x,y ~ P(x,y) \to Q(f(y),y,y)$ and $\exists g \forall x,y,z ~ Q(x,y,z) \to Q(x,y,g(x,y))$.  Then $\sm{M}^{\rm skel}$ will include the s-t tgd $P(x,y) \to Q_{f(\bullet),\bullet,\bullet}(x,y,y)$, and the target tgds $Q_{\bullet,\bullet,\bullet}(x,y,z) \to Q_{\bullet,\bullet,g(\bullet,\bullet)}(x,y,x,y)$, and $Q_{f(\bullet),\bullet,\bullet}(x,y,z) \to Q_{f(\bullet),\bullet,g(f(\bullet),\bullet)}(x,y,x,y)$.
The full skeleton rewriting of our running example schema mapping is given in Figure~\ref{fig:skeletonized}.  

\begin{figure*}[!t]
	\centering
	\begin{tabular}{|lll|}
		\hline
		$\bf S$\hspace{-2mm} & $=$ & $\left\{ \rel R \right\}$ \\
		$\bf T$\hspace{-2mm} & $=$ & $\left\{ \rel T_{\bullet,\bullet},\rel T_{\bullet,f(\bullet,\bullet)},\rel T_{f(\bullet,\bullet),\bullet},\rel T_{f(\bullet,\bullet),f(\bullet,\bullet)},\rel{Eq}_{\bullet,\bullet},\rel{Eq}_{\bullet,f(\bullet,\bullet)},\rel{Eq}_{f(\bullet,\bullet),\bullet},\rel{Eq}_{f(\bullet,\bullet),f(\bullet,\bullet)} \right\}$ \\
		\sigst\hspace{-2mm} & $=$ & $\left\{ \begin{array}{lllll}
			\rel R(x,y) & \to & \rel T_{\bullet,f(\bullet,\bullet)}(x,x,y)\\
			\rel R(x,y) & \to & \rel T_{\bullet,f(\bullet,\bullet)}(y,x,y)
		\end{array} \right\}$ \\
		$\sigt^{\rm skel}$\hspace{-2mm} & $=$ & \hspace{-2mm}$\left\{\hspace{-1mm} \begin{array}{l}
			\rel T_{\bullet,\bullet}(x,y) \wedge \rel{Eq}_{\bullet,\bullet}(x,x') \wedge \rel T_{\bullet,\bullet}(x',y') \to \rel{Eq}_{\bullet,\bullet}(y,y')\\
			\rel T_{\bullet,\bullet}(x,y) \wedge \rel{Eq}_{\bullet,\bullet}(x,x') \wedge \rel T_{\bullet,f(\bullet,\bullet)}(x',y'_1,y'_2) \to \rel{Eq}_{\bullet,f(\bullet,\bullet)}(y,y'_1,y'_2)\\
			\rel T_{\bullet,\bullet}(x,y) \wedge \rel{Eq}_{\bullet,f(\bullet,\bullet)}(x,x'_1,x'_2) \wedge \rel T_{f(\bullet,\bullet),\bullet}(x'_1,x'_2,y') \to \rel{Eq}_{\bullet,\bullet}(y,y')\\
			\rel T_{\bullet,\bullet}(x,y) \wedge \rel{Eq}_{\bullet,f(\bullet,\bullet)}(x,x'_1,x'_2) \wedge \rel T_{f(\bullet,\bullet),f(\bullet,\bullet)}(x'_1,x'_2,y'_1,y'_2) \to \rel{Eq}_{\bullet,f(\bullet,\bullet)}(y,y'_1,y'_2)\\
			\rel T_{\bullet,f(\bullet,\bullet)}(x,y_1,y_2) \wedge \rel{Eq}_{\bullet,\bullet}(x,x') \wedge \rel T_{\bullet,\bullet}(x',y') \to \rel{Eq}_{f(\bullet,\bullet),\bullet}(y_1,y_2,y')\\
			\rel T_{\bullet,f(\bullet,\bullet)}(x,y_1,y_2) \wedge \rel{Eq}_{\bullet,\bullet}(x,x') \wedge \rel T_{\bullet,f(\bullet,\bullet)}(x',y'_1,y'_2) \to \rel{Eq}_{f(\bullet,\bullet),f(\bullet,\bullet)}(y_1,y_2,y'_1,y'_2)\\
			\rel T_{\bullet,f(\bullet,\bullet)}(x,y_1,y_2) \wedge \rel{Eq}_{\bullet,f(\bullet,\bullet)}(x,x'_1,x'_2) \wedge \rel T_{f(\bullet,\bullet),\bullet}(x'_1,x'_2,y') \to \rel{Eq}_{f(\bullet,\bullet),\bullet}(y_1,y_2,y')\\
			\rel T_{\bullet,f(\bullet,\bullet)}(x,y_1,y_2) \wedge \rel{Eq}_{\bullet,f(\bullet,\bullet)}(x,x'_1,x'_2) \wedge \rel T_{f(\bullet,\bullet),f(\bullet,\bullet)}(x'_1,x'_2,y'_1,y'_2) \to \rel{Eq}_{f(\bullet,\bullet),f(\bullet,\bullet)}(y_1,y_2,y'_1,y'_2)\\

			\rel T_{f(\bullet,\bullet),\bullet}(x_1,x_2,y) \wedge \rel{Eq}_{f(\bullet,\bullet),\bullet}(x_1,x_2,x') \wedge \rel T_{\bullet,\bullet}(x',y') \to \rel{Eq}_{\bullet,\bullet}(y,y')\\
			\rel T_{f(\bullet,\bullet),\bullet}(x_1,x_2,y) \wedge \rel{Eq}_{f(\bullet,\bullet),\bullet}(x_1,x_2,x') \wedge \rel T_{\bullet,f(\bullet,\bullet)}(x',y'_1,y'_2) \to \rel{Eq}_{\bullet,f(\bullet,\bullet)}(y,y'_1,y'_2)\\
			\rel T_{f(\bullet,\bullet),\bullet}(x_1,x_2,y) \wedge \rel{Eq}_{f(\bullet,\bullet),f(\bullet,\bullet)}(x_1,x_2,x'_1,x'_2) \wedge \rel T_{f(\bullet,\bullet),\bullet}(x'_1,x'_2,y') \to \rel{Eq}_{\bullet,\bullet}(y,y')\\
			\rel T_{f(\bullet,\bullet),\bullet}(x_1,x_2,y) \wedge \rel{Eq}_{f(\bullet,\bullet),f(\bullet,\bullet)}(x_1,x_2,x'_1,x'_2) \wedge \rel T_{f(\bullet,\bullet),f(\bullet,\bullet)}(x'_1,x'_2,y'_1,y'_2) \to \rel{Eq}_{\bullet,f(\bullet,\bullet)}(y,y'_1,y'_2)\\
			\rel T_{f(\bullet,\bullet),f(\bullet,\bullet)}(x_1,x_2,y_1,y_2) \wedge \rel{Eq}_{f(\bullet,\bullet),\bullet}(x_1,x_2,x') \wedge \rel T_{\bullet,\bullet}(x',y') \to \rel{Eq}_{f(\bullet,\bullet),\bullet}(y_1,y_2,y')\\
			\rel T_{f(\bullet,\bullet),f(\bullet,\bullet)}(x_1,x_2,y_1,y_2) \wedge \rel{Eq}_{f(\bullet,\bullet),\bullet}(x_1,x_2,x') \wedge \rel T_{\bullet,f(\bullet,\bullet)}(x',y'_1,y'_2) \to \rel{Eq}_{f(\bullet,\bullet),f(\bullet,\bullet)}(y_1,y_2,y'_1,y'_2)\\
			\rel T_{f(\bullet,\bullet),f(\bullet,\bullet)}(x_1,x_2,y_1,y_2) \wedge \rel{Eq}_{f(\bullet,\bullet),f(\bullet,\bullet)}(x_1,x_2,x'_1,x'_2) \wedge \rel T_{f(\bullet,\bullet),\bullet}(x'_1,x'_2,y') \to \rel{Eq}_{f(\bullet,\bullet),\bullet}(y_1,y_2,y')\\
			\rel T_{f(\bullet,\bullet),f(\bullet,\bullet)}(x_1,x_2,y_1,y_2) \wedge \rel{Eq}_{f(\bullet,\bullet),f(\bullet,\bullet)}(x_1,x_2,x'_1,x'_2) \wedge \rel T_{f(\bullet,\bullet),f(\bullet,\bullet)}(x'_1,x'_2,y'_1,y'_2) \to \\
			\hfill\rel{Eq}_{f(\bullet,\bullet),f(\bullet,\bullet)}(y_1,y_2,y'_1,y'_2)\\

			\rel T_{\bullet,\bullet}(x,y) \to \rel{Eq}_{\bullet,\bullet}(x,x)\\
			\rel T_{\bullet,\bullet}(x,y) \to \rel{Eq}_{\bullet,\bullet}(y,y)\\
			\rel T_{\bullet,f(\bullet,\bullet)}(x,y_1,y_2) \to \rel{Eq}_{\bullet,\bullet}(x,x)\\
			\rel T_{\bullet,f(\bullet,\bullet)}(x,y_1,y_2) \to \rel{Eq}_{f(\bullet,\bullet),f(\bullet,\bullet)}(y_1,y_2,y_1,y_2)\\
			\rel T_{f(\bullet,\bullet),\bullet}(x_1,x_2,y) \to \rel{Eq}_{f(\bullet,\bullet),f(\bullet,\bullet)}(x_1,x_2,x_1,x_2)\\
			\rel T_{f(\bullet,\bullet),\bullet}(x_1,x_2,y) \to \rel{Eq}_{\bullet,\bullet}(y,y)\\
			\rel T_{f(\bullet,\bullet),f(\bullet,\bullet)}(x_1,x_2,y_1,y_2) \to \rel{Eq}_{f(\bullet,\bullet),f(\bullet,\bullet)}(x_1,x_2,x_1,x_2)\\
			\rel T_{f(\bullet,\bullet),f(\bullet,\bullet)}(x_1,x_2,y_1,y_2) \to \rel{Eq}_{f(\bullet,\bullet),f(\bullet,\bullet)}(y_1,y_2,y_1,y_2)\\

			\rel{Eq}_{\bullet,\bullet}(x_1,x_2) \to \rel{Eq}_{\bullet,\bullet}(x_2,x_1)\\
			\rel{Eq}_{\bullet,f(\bullet,\bullet)}(x_1,x_{2_1},x_{2_2}) \to \rel{Eq}_{f(\bullet,\bullet),\bullet}(x_{2_1},x_{2_2},x_1)\\
			\rel{Eq}_{f(\bullet,\bullet),\bullet}(x_{1_1},x_{1_2},x_2) \to \rel{Eq}_{\bullet,f(\bullet,\bullet)}(x_2,x_{1_1},x_{1_2})\\
			\rel{Eq}_{f(\bullet,\bullet),f(\bullet,\bullet)}(x_{1_1},x_{1_2},x_{2_1},x_{2_2}) \to \rel{Eq}_{f(\bullet,\bullet),f(\bullet,\bullet)}(x_{2_1},x_{2_2},x_{1_1},x_{1_2})\\

			\rel{Eq}_{\bullet,\bullet}(x_1,x_2) \wedge \rel{Eq}_{\bullet,\bullet}(x_2,x_3) \to \rel{Eq}_{\bullet,\bullet}(x_1,x_3)\\
			\rel{Eq}_{\bullet,\bullet}(x_1,x_2) \wedge \rel{Eq}_{\bullet,f(\bullet,\bullet)}(x_2,x_{3_1},x_{3_2}) \to \rel{Eq}_{\bullet,f(\bullet,\bullet)}(x_1,x_{3_1},x_{3_2})\\
			\rel{Eq}_{\bullet,f(\bullet,\bullet)}(x_1,x_{2_1},x_{2_2}) \wedge \rel{Eq}_{f(\bullet,\bullet),\bullet}(x_{2_1},x_{2_2},x_3) \to \rel{Eq}_{\bullet,\bullet}(x_1,x_3)\\
			\rel{Eq}_{\bullet,f(\bullet,\bullet)}(x_1,x_{2_1},x_{2_2}) \wedge \rel{Eq}_{f(\bullet,\bullet),f(\bullet,\bullet)}(x_{2_1},x_{2_2},x_{3_1},x_{3_2}) \to \rel{Eq}_{\bullet,f(\bullet,\bullet)}(x_1,x_{3_1},x_{3_2})\\
			\rel{Eq}_{f(\bullet,\bullet),\bullet}(x_{1_1},x_{1_2},x_2) \wedge \rel{Eq}_{\bullet,\bullet}(x_2,x_3) \to \rel{Eq}_{f(\bullet,\bullet),\bullet}(x_{1_1},x_{1_2},x_3)\\
			\rel{Eq}_{f(\bullet,\bullet),\bullet}(x_{1_1},x_{1_2},x_2) \wedge \rel{Eq}_{\bullet,f(\bullet,\bullet)}(x_2,x_{3_1},x_{3_2}) \to \rel{Eq}_{f(\bullet,\bullet),f(\bullet,\bullet)}(x_{1_1},x_{1_2},x_{3_1},x_{3_2})\\
			\rel{Eq}_{f(\bullet,\bullet),f(\bullet,\bullet)}(x_{1_1},x_{1_2},x_{2_1},x_{2_2}) \wedge \rel{Eq}_{f(\bullet,\bullet),\bullet}(x_{2_1},x_{2_2},x_3) \to \rel{Eq}_{f(\bullet,\bullet),\bullet}(x_{1_1},x_{1_2},x_3)\\
			\rel{Eq}_{f(\bullet,\bullet),f(\bullet,\bullet)}(x_{1_1},x_{1_2},x_{2_1},x_{2_2}) \wedge \rel{Eq}_{f(\bullet,\bullet),f(\bullet,\bullet)}(x_{2_1},x_{2_2},x_{3_1},x_{3_2}) \to \rel{Eq}_{f(\bullet,\bullet),f(\bullet,\bullet)}(x_{1_1},x_{1_2},x_{3_1},x_{3_2})\\
			
		\end{array} \hspace{-2mm}\right\}\hspace{-1mm}$\\
		\multicolumn{3}{|l|}{$\begin{array}{rcl}
				\rel{reachable}^{\rm skel}(x,y) & \defeq & \exists c,c' (\rel T_{\bullet,\bullet}(x, c) \wedge \rel{Eq}_{\bullet,\bullet}(c,c') \wedge \rel T_{\bullet,\bullet}(y,c'))\\
				& \cup & \exists c,c'_1,c'_2 (\rel T_{\bullet,\bullet}(x, c) \wedge \rel{Eq}_{\bullet,f(\bullet,\bullet)}(c,c'_1,c'_2) \wedge \rel T_{\bullet,f(\bullet,\bullet)}(y,c'_1,c'_2))\\
				& \cup & \exists c_1,c_2,c' (\rel T_{\bullet,f(\bullet,\bullet)}(x, c_1,c_2) \wedge \rel{Eq}_{f(\bullet,\bullet),\bullet}(c_1,c_2,c') \wedge \rel T_{\bullet,\bullet}(y,c'))\\
				& \cup & \exists c_1,c_2,c'_1,c'_2 (\rel T_{\bullet,f(\bullet,\bullet)}(x, c_1,c_2) \wedge \rel{Eq}_{f(\bullet,\bullet),f(\bullet,\bullet)}(c_1,c_2,c'_1,c'_2) \wedge \rel T_{\bullet,f(\bullet,\bullet)}(y,c'_1,c'_2))
			\end{array}$}\\
		\hline
	\end{tabular}
    \caption{Undirected reachability, skolemized, equality singularized, and skeleton rewritten.}
    \label{fig:skeletonized}
\end{figure*}

\begin{remark}
An optimized version of the schema mapping in Figure~\ref{fig:skeletonized} is shown in Figure~\ref{fig:optimized}, based on the simple observation that in Figure~\ref{fig:skeletonized} none of $\{\rel T_{\bullet,\bullet},\rel T_{f(\bullet,\bullet),\bullet},\rel T_{f(\bullet,\bullet),f(\bullet,\bullet)} \}$ appears on the right-hand side of any tgd, and thus the left-hand sides of many tgds cannot be satisfied in any universal solution, and in turn none of $\{\rel{Eq}_{\bullet,f(\bullet,\bullet)},\rel{Eq}_{f(\bullet,\bullet),\bullet}\}$ ever appears on the right-hand side of a remaining tgd in which it does not also appear on the left-hand side.  We leave development of a principled approach to optimization for future work.
\end{remark}

\begin{figure*}[htbp]
	\centering
	\begin{tabular}{|lll|}
		\hline
		$\bf S$ & $=$ & $\left\{ \rel R \right\}$ \\
		$\bf T$ & $=$ & $\left\{ \rel T_{\bullet,f(\bullet,\bullet)}, \rel{Eq}_{\bullet,\bullet},\rel{Eq}_{f(\bullet,\bullet),f(\bullet,\bullet)} \right\}$ \\
		\sigst & $=$ & $\left\{ \begin{array}{lllll}
			\rel R(x,y) & \to & \rel T_{\bullet,f(\bullet,\bullet)}(x,x,y)\\
			\rel R(x,y) & \to & \rel T_{\bullet,f(\bullet,\bullet)}(y,x,y)
		\end{array} \right\}$ \\
		$\sigt^{\rm skel}$ & $=$ & $\left\{ \begin{array}{l}
			\rel T_{\bullet,f(\bullet,\bullet)}(x,y_1,y_2) \wedge \rel{Eq}_{\bullet,\bullet}(x,x') \wedge \rel T_{\bullet,f(\bullet,\bullet)}(x',y'_1,y'_2) \to \rel{Eq}_{f(\bullet,\bullet),f(\bullet,\bullet)}(y_1,y_2,y'_1,y'_2)\\

			\rel T_{\bullet,f(\bullet,\bullet)}(x,y_1,y_2) \to \rel{Eq}_{\bullet,\bullet}(x,x) \\
			\rel T_{\bullet,f(\bullet,\bullet)}(x,y_1,y_2) \to \rel{Eq}_{f(\bullet,\bullet),f(\bullet,\bullet)}(y_1,y_2,y_1,y_2)\\

			\rel{Eq}_{\bullet,\bullet}(x_1,x_2) \to \rel{Eq}_{\bullet,\bullet}(x_2,x_1)\\
			\rel{Eq}_{f(\bullet,\bullet),f(\bullet,\bullet)}(x_{1_1},x_{1_2},x_{2_1},x_{2_2}) \to \rel{Eq}_{f(\bullet,\bullet),f(\bullet,\bullet)}(x_{2_1},x_{2_2},x_{1_1},x_{1_2})\\

			\rel{Eq}_{\bullet,\bullet}(x_1,x_2) \wedge \rel{Eq}_{\bullet,\bullet}(x_2,x_3) \to \rel{Eq}_{\bullet,\bullet}(x_1,x_3)\\
			\rel{Eq}_{f(\bullet,\bullet),f(\bullet,\bullet)}(x_{1_1},x_{1_2},x_{2_1},x_{2_2}) \wedge \rel{Eq}_{f(\bullet,\bullet),f(\bullet,\bullet)}(x_{2_1},x_{2_2},x_{3_1},x_{3_2}) \to \rel{Eq}_{f(\bullet,\bullet),f(\bullet,\bullet)}(x_{1_1},x_{1_2},x_{3_1},x_{3_2})\\
			
		\end{array} \right\}$\\
		\multicolumn{3}{|l|}{$\begin{array}{rcl}
				\rel{reachable}^{\rm skel}(x,y) & \defeq & \exists c_1,c_2,c'_1,c'_2 (\rel T_{\bullet,f(\bullet,\bullet)}(x, c_1,c_2) \wedge \rel{Eq}_{f(\bullet,\bullet),f(\bullet,\bullet)}(c_1,c_2,c'_1,c'_2) \wedge \rel T_{\bullet,f(\bullet,\bullet)}(y,c'_1,c'_2))
			\end{array}$}\\
		\hline
	\end{tabular}
    \caption{Example schema mapping from Figure~\ref{fig:example}, skolemized, equality singularized, skeleton rewritten, and optimized.}
    \label{fig:optimized}
\end{figure*}

\begin{proposition} \label{prp:skel_same}
	Let \mapping be a weakly acyclic \mappingtypetwo{\efsotgd}{\efsotgd} schema mapping.  Let ${\cal M}^{\rm skel}$ be the skeleton rewriting of \sm{M}.
	For every UCQ $q$ over $\bf T$ and for every $\bf S$-instance $I$, we have $\certain(q,I,M) = \certain(q^{\rm skel},I,M')$.
\end{proposition}

\begin{proof}[Hint]
	We show that there exists a solution $J$ for $I$ with respect to $\sm{M}$ if and only if there exists a solution 
	$J'$ for $I$ with respect to ${\cal M}^{\rm skel}$.  Furthermore, $J'$ (respectively $J$) can be constructed such that for any UCQ $q$ over $\bf T$, we have $q{\downarrow}(J) = q^{\rm skel}{\downarrow}(J')$.
	To construct $J$ from $J'$, we copy every tuple, and use the skeletons and their arguments to construct the compound terms.
	To construct $J'$ from $J$, we copy every tuple, and, using a witnessing collection of functions ${\bf f^0}$ such that $(I,J)\models \sm{M}~[{\bf f}\mapsto{\bf f_0}]$, and such that each null value is the denotation of a unique term of depth at most $r$. This term gives us both the skeleton and the arguments that belong in $J'$.
\end{proof}

\subsection{Proof of Theorem~\ref{thm:sotgdtogav}}

We finally can prove Theorem~\ref{thm:sotgdtogav} by combining the above results:
	Let \mapping be a weakly acyclic \mappingtypethree{\sotgd}{\sotgd}{\egd} schema mapping, 
	and let \sm{\hat{M}} be the skeleton rewriting of the equality singularization of \sm{M},
	extended with the egd \[\rel{Eq}_{\bullet,\bullet}(x,y)\to x=y~.\]
	Furthermore, for any UCQ $q$ over ${\bf T}$, let $\hat{q}$ be the skeleton rewriting of the  
	equality simulation of $q$.  Then we claim that $\xrcertain(q,I,M) = \xrcertain(\hat{q},I,\sm{\hat{M}})$.

It suffices to show that, for all source instances $I$, 
\begin{enumerate}
\item $I$ has a solution with respect to $\sm{\hat{M}}$ if and only if $I$ has a solution with respect to $\sm{M}$.
\item if $I$ has a solution with respect to $\sm{M}$, then, for all UCQs $q$ over ${\bf T}$, $\certain(\hat{q},I,\sm{\hat{M}}) = \certain(q,I,\sm{M})$
\end{enumerate}

The first item follows from Proposition~\ref{prop:singularization}(a) and Proposition~\ref{prp:skel_same}.
The second item follows from Proposition~\ref{prop:singularization}(b) and Proposition~\ref{prp:skel_same}.

\subsection{Related Work}
Theorem~\ref{thm:sotgdtogav} allows us to extend the DLP-rewriting technique of Section~\ref{sec:dlp} to \mappingtypethree{\glav}{\waglav}{\egd} schema mappings (and, in fact, to weakly acyclic \mappingtypethree{\sotgd}{\sotgd}{\egd} schema mappings). The proof is based on a method for eliminating the existentially quantified variables.  Others have considered methods for eliminating existential quantifiers from tgds previously, an early example being Duschka and Genesereth's inverse rules algorithm~\cite{DBLP:conf/pods/DuschkaG97} for acyclic \lav{} rules, which inspired our approach.  Krotzsch and Rudolph describe an existentially quantified variable elimination procedure for schema mappings composed of \glav{} constraints and relational denial constraints (a subset of denial constraints with no equality or inequality atoms) that are jointly-acyclic (a relaxation of weak acyclicity) in~\cite{DBLP:conf/ijcai/KrotzschR11}.  Their approach is similar to ours in that it creates extra attributes to represent skolem terms in place of existentially quantified variables, but our constraint language includes the additional expressiveness of egds, whose careful handling is a primary concern of our approach.  Marnette studied termination of the chase for schema mappings with target constraints in~\cite{DBLP:conf/pods/Marnette09}, where he introduced the oblivious skolem chase, a modification of the chase procedure in which skolem terms are allowed to appear in instances.  A similar procedure was used to prove the correctness of a limited form of skeleton rewriting in~\cite{DBLP:conf/rr/ten-CateHK14}.

Equality singularization for tgds was introduced in~\cite{DBLP:conf/pods/Marnette09}, where it was referred to simply as ``singularization''.
In~\cite{DBLP:conf/rr/ten-CateHK14}, another equality simulation technique was used, based on substitution.  In that presentation, the simulation was woven into the skeleton rewriting step.

Theorem~\ref{thm:sotgdtogav} is related to a result in an unpublished manuscript \cite{DBLP:journals/corr/abs-1212-0254}, which can be stated as follows: given any \mappingtypetwo{\glav}{\waglav} schema mapping \sm{M} and every conjunctive query $q$, one can compute a Datalog program that, given any source instance as input, computes the certain answers of $q$ with respect to \sm{M}. Note that, conceptually, a Datalog program can be viewed as a \mappingtypetwo{\gav}{\gav} schema mapping where the source schema consists of the EDB predicates and the target schema consists of the IDB predicates.

\section{Concluding Remarks}
In this paper, we introduced the framework of exchange-repairs and explored the XR-certain answers as an alternative non-trivial and meaningful semantics of queries in the context of data exchange.  Exchange-repair semantics differ from other proposals for handling inconsistencies in data exchange in that, conceptually, the inconsistencies are repaired at the source rather than the target.  This allows the shared origins of target facts to be reflected in the answers to target queries. 

This framework brings together data exchange, database repairs, and disjunctive logic programming, thus enhancing the interaction between three different areas of research. Moreover, the results reported here pave the way for using DLP solvers, such as DLV, for query answering under the exchange-repair semantics.

\section{Acknowledgements}
The research of all authors was partially supported by NSF Grant IIS-1217869. Kolaitis' research was also supported by the project ``Handling Uncertainty in Data Intensive Applications" under the program THALES.

%
%
\bibliographystyle{habbrv}
\bibliography{Exchange-Repair}

\appendix

\end{document}